\documentclass[11pt]{article}
\usepackage{amsmath}
\usepackage{amsthm}
\usepackage{amsfonts}
\usepackage{amssymb}
\usepackage{mathtools}
\usepackage[numbers,sort&compress]{natbib}

\usepackage{cases}
\usepackage{xcolor}

\newcommand{\dt}{\partial_t}
\newcommand{\dv}{\mathrm{div}\,}
\newcommand{\dvh}{\mathrm{div}_h\,}

\newcommand{\idx}[0]{\,dx}
\newcommand{\subeqref}[2]{$\eqref{#1}_{#2}$}

\newcommand{\norm}[2]{\Arrowvert #1 \Arrowvert_{#2}}

\newcommand{\Lnorm}[1]{L^{#1}(\Omega)}
\newcommand{\Hnorm}[1]{H^{#1}(\Omega)}
\newcommand{\bLnorm}[1]{L^{#1}(\Gamma)}
\newcommand{\bHnorm}[1]{H^{#1}(\Gamma)}

\renewcommand{\vec}[1]{\mathbf{#1}}

\newcommand{\mrm}[1]{\mathrm{#1}}
\newcommand{\pmindex}{{\mrm j \in \lbrace+,-\rbrace}}

\newtheorem{thm}{Theorem}[section]

\theoremstyle{remark}
\newtheorem{remark}{Remark}

\numberwithin{equation}{section}
\allowdisplaybreaks
\title{On the resolvability of the dynamics of one fluid flow via a testing fluid in a two-fluids flow model}

\author{Xin Liu
\footnote{Weierstrass-Institut
f\"ur Angewandte Analysis und Stochastik,
Leibniz-Institut im Forschungsverbund Berlin, Berlin Germany.
Email: stleonliu@gmail.com\, and \, stleonliu@live.com.}}

\begin{document}
\allowdisplaybreaks
\maketitle

\begin{abstract}
	In this work, by considering an isentropic fluid-fluid interaction model with a large symmetric drag force, a commonly used simplified two-fluids flow model is justified as the asymptotic limit. Equations for each fluid component with an interaction term are identified in addition to the simplified two-fluids flow model, which can be used to resolve the density of one fluid specie based on information on the density and the velocity of the other fluid specie, i.e., the testing flow. 
	{\par\noindent\bf Keywords:} Asymptotic limit; Two-fluids flow model; Fluid-fluid interaction.
\end{abstract}

\tableofcontents

\section{Introduction}

The goal of this paper is to investigate the asymptotic limit $ \sigma \rightarrow \infty $ of system
\begin{equation}\label{eq:2CMPNS}
	\begin{cases}
		\dt \rho_{\sigma,\pm} + \dv (\rho_{\sigma,\pm} \vec u_{\sigma,\pm} ) = 0, \\
		\dt (\rho_{\sigma,\pm} \vec u_{\sigma,\pm} ) + \dv (\rho_{\sigma,\pm} \vec u_{\sigma,\pm} \otimes \vec u_{\sigma,\pm} ) + \nabla p_{\sigma,\pm} = \dv \mathbb S_{\sigma,\pm} \\
		\qquad + \sigma \rho_{\sigma,\pm} \rho_{\sigma,\mp} (\vec u_{\sigma,\mp} - \vec u_{\sigma,\pm}),
	\end{cases}
\end{equation}
in $ \Omega = \mathbb T^3 $ or $ \mathbb T^2\times (0,1) $ with proper boundary conditions to be specified, 
where
\begin{equation}\label{def:pressure-viscosity}
\begin{gathered}
	p_{\sigma,\pm} = R_\pm \rho_{\sigma,\pm}^{\gamma_\pm} , \\
	\mathbb S_{\sigma,\pm} = \mu_\pm (\nabla \vec u_{\sigma,\pm} + \nabla^\top \vec u_{\sigma,\pm}) + \lambda_\pm  \dv \vec u_{\sigma,\pm} \mathbb I_3.
\end{gathered} 
\end{equation}
Here 
\begin{equation}\label{asm:general}
	R_\pm > 0, ~ \gamma_\pm > 1, ~ \mu_\pm > 0, ~ \dfrac{2}{3} \mu_\pm + \lambda_\pm > 0, ~ \sigma > 0,
\end{equation}
are constants, representing the gas constants, the adiabatic constants, the shear viscosity coefficients, the bulk viscosity coefficients, and the interaction coefficients, respectively. 

We verify rigorously that, the asymptotic limit of solutions $ (\rho_{\sigma,\pm},\vec u_{\sigma,\pm}) $ to \eqref{eq:2CMPNS}, as $ \sigma \rightarrow \infty $, will converge to $ (\rho_\pm, \vec u) $, which is the solution to the following two-fluids flow model,
\begin{equation}\label{eq:limiting-eq}
	\begin{cases}
	\dt \rho_\pm + \dv (\rho_{\pm} \vec u) = 0,\\
	\dt \bigl((\rho_+ + \rho_-) \vec u \bigr) + \dv \bigl((\rho_+ + \rho_-) \vec u \otimes \vec u \bigr) + \nabla p = \dv \mathbb S,
	\end{cases}
\end{equation}
where, formally,
\begin{gather*}
	\rho_{\pm} := \lim\limits_{\sigma \rightarrow \infty} \rho_{\sigma, \pm}, \quad \vec u_{\pm} := \lim\limits_{\sigma \rightarrow \infty} \vec u_{\sigma, \pm},
\end{gather*}
under some appropriate assumptions, and
\begin{equation}\label{def:pressure-viscosity-limiting}
	\begin{gathered}
		p = \sum_{\pmindex} R_\mrm j \rho_\mrm j^{\gamma_\mrm j} , \\
		\mathbb S = \sum_{\pmindex} (\mu_\mrm j(\nabla \vec u + \nabla^\top \vec u) + \lambda_\mrm j \dv \vec u \mathbb{I}_3).
	\end{gathered}
\end{equation}

As a byproduct, we discover a method to resolve the density of one fluid from the density and the velocity of the other fluid, i.e., the testing fluid, in the mixture of two fluids. More precisely, the equations,
\begin{equation}\label{eq:conv-3}
	\begin{gathered}
	\dt (\rho_{\pm}\vec u) + \dv (\rho_{\pm}\vec u\otimes \vec u) + \nabla (R_\pm \rho_{\pm}^{\gamma_\pm })\\
	 = \dv(\mu_\pm (\nabla \vec u + \nabla^\top \vec u) + \lambda_\pm \dv \vec u\mathbb I_3 ) 
	 \pm \rho_{\pm}\rho_{\mp} \mathfrak S,
	 \end{gathered}
\end{equation}
are identified, where $ \mathfrak S $ is a quantity relating the two densities $ \rho_+, \rho_- $ and the velocity $\vec u $ (see \eqref{def:interaction}, below). Thus, provided with information of the testing fluid, $(\rho_+,\vec u)$ for instance, one can calculate $ \rho_- \mathfrak S $ from \eqref{eq:conv-3}. From there, one can identify the density of the other flow, i.e., $\rho_-$ in this example.

System \eqref{eq:limiting-eq} is a simplified two-fluids flow model. The idea of multi-components of the fluid sharing the same aligned velocity is commonly used in a lot of applications. For instance, in the study of atmosphere dynamics, dry air, water vapor, and cloud water are driven by the same velocity, which is determined by a single momentum equation similar to \eqref{eq:limiting-eq}. See, e.g., \cite{Hittmeir2018,Klein2006}. This system also serves as a plausible model for the study of particle/fluid interaction. See, e.g., \cite{Mellet2008,Evje2011}, and the references therein. 

In particular, in \cite{Mellet2008}, starting from the Vlasov-Fokker-Planck/Navier-Stokes system, the authors investigate the asymptotic with a strong drag force and a strong Brownian motion, where the limiting system is similar to \eqref{eq:limiting-eq}. However, the drag force is taken to be asymmetric. In fact, as explained in Remark 1.1 in \cite{Mellet2008}, the drag force was taken as $ F_d = F_0(\vec u- \vec v) $, where $ F_0 $ is a constant, and the more physical relevant one should be taken as $ F_d = \rho (\vec u - \vec v) $, which we will refer to as the symmetric drag force. 

Formally taking the hydrodynamic limit of the Vlasov-Fokker-Planck/Navier-Stokes system, with a symmetric drag force, one will end up with system \eqref{eq:2CMPNS}, while in the case of an asymmetric drag force in the kinetic-hydrodynamic system, it leads to a two-phase fluid model similar to \eqref{eq:2CMPNS}  but with the drag force equal to $ \sigma \rho_{\sigma,+}(\vec u_{\sigma,\mp} - \vec u_{\sigma,\pm}) $, i.e., an asymmetric drag force. See \cite{ChoiSIAM2016} for the formal derivation of the hydrodynamic limit. We will use the terminology of symmetric or asymmetric drag forces for hydrodynamic systems as well as kinetic-hydrodynamic systems.

In addition, from the point of view of hydrodynamics, system \eqref{eq:2CMPNS} can be seen as a model of the mixture of two fluids, while each one of the fluids acts as a porous media to the other, and neither of the fluids is dominating the other. Therefore, the drag force should be symmetric.

Due to the low regularity of solutions studied in \cite{Mellet2008}, the lower bound of density is not {\it a priori} known. Therefore only large asymmetric drag force was considered. In this work, we want to investigate the large drag force limit in the setting of hydrodynamic systems, with the symmetric drag force, in a more regular functional setting. We remark that it would be interesting to investigate the large, symmetric drag force limit in the setting of kinetic-hydrodynamic systems. 

Among a large amount of literatures concerning two-phase fluid models, we will only mention a few in the following, and refer interested readers to the references therein. Modeled by the Vlasov-Fokker-Planck/Navier-Stokes system, Mellet and Vasseur construct the global weak solution in \cite{Vasseur2007M3AS}, while the asymptotic limit of a large drag force and a large Brownian motion is investigated in \cite{Mellet2008}. Without fluid viscosities, the local well-posedness of strong solutions is studied in \cite{Baranger2006}. Global classical solutions near equilibrium, as well as the decaying rates of perturbations, are studied in \cite{Chae2013,LiMuWangSIAM2017} in the present of fluid viscosities. 

Concerning the limiting equations in \cite{Mellet2008}, the existence of global weak solutions is established in \cite{Vasseur2019}. See \cite{Wen2019} for the recent improvement. We also refer to \cite{Evje2017,Evje2011}, as well as the references therein, for early mathematical developments as well as the physical importance of this two-fluids flow model. 

Near equilibrium but with an asymmetric drag force, global existence of classical solutions is investigated in \cite{ChoiSIAM2016}. We refer readers to \cite{Bresch2010,BreschHuangLiCMP2012,Novotny2020,Evje2016,Bresch2011} for a more general two-fluids flow model, where volumetric rates are taken into account, and related studies.

On the other hand, the mathematical study of mono-fluid, i.e., with one specie of fluid, has developed fruitful results. To name a few, starting with \cite{Feireisl2004,Lions1998}, the authors construct the well-known Lions-Feireisl weak solutions to compressible Navier-Stokes equations. The theory of local well-posedness of strong solutions with vacuum density profiles is established in \cite{Cho2004,Cho2006a,Cho2006c}. Global existence in the framework of perturbation is obtained in \cite{Matsumura1980,Matsumura1983}. The blow-up and non-existence to full compressible Navier-Stokes system with vacuum and bounded entropy are shown in \cite{XinYan2013,zpxin1998,Li2017b}. We also refer readers to \cite{Bresch2007,Mellet2007,Vasseur2016,HuangLiXin2012,Li2015a} for other important developments in this direction. 

The rest of this paper is organized as follows. In the next section, we collect the notations we will be used in this paper and state the main theorems. Section \ref{sec:t3} and section \ref{sec:flat} are devoted to establish the uniform-in-$\sigma$ estimates, which will be the centerpiece of our analysis. In section \ref{sec:limit}, we pass the limit $ \sigma \rightarrow \infty $, which yields the results in this paper. 

\section{Preliminaries and main theorems}

We investigate our problem in domain $ \Omega \subset \mathbb R^3 $, where
\begin{equation*}
	\Omega = \lbrace (x,y,z)^\top \in \mathbb T^3 \rbrace \quad \text{or} \quad \lbrace (x,y,z)^\top | (x,y)^\top \in \mathbb T^2, z \in (0,1) \rbrace.
\end{equation*}
$ \nabla $, $ \dv $, and $ \Delta $ represent the gradient, the divergence, and the Laplace operators, respectively. Meanwhile, $ \nabla_h $, $ \dv_h $, and $ \Delta_h $ are the gradient, the divergence, and the Laplace operators in the horizontal (first two) variables, respectively, i.e.,
\begin{equation*}
	\nabla_h := \left( \begin{array}{c}
		\partial_x \\ \partial_y
	\end{array}\right), \quad \dvh := \nabla_h \cdot, \quad \text{and} \quad \Delta_h := \dvh \nabla_h. 
\end{equation*}

We use $ \norm{\cdot}{X} $ to denote the norm of functional space $ X $. We shorten the notation $ \norm{f}{X} + \norm{g}{X} + \cdots $ to $ \norm{f,g,\cdots}{X} $ for norms of multiple functions.
$ \int f \idx  := \int_\Omega f \idx $
represents the integration in the spatial variables. 

$ (\rho_{\sigma,\pm,0},\vec u_{\sigma,\pm,0}) $ are used to represent the initial data of \eqref{eq:2CMPNS}, i.e.,
\begin{equation}\label{def:initial}
	(\rho_{\sigma,\pm},\vec u_{\sigma,\pm})\big|_{t = 0} = (\rho_{\sigma,\pm,0},\vec u_{\sigma,\pm,0}).
\end{equation}
Let 
\begin{equation}\label{def:lowerbound-rho}
	\underline{\rho} := \inf_{x\in\Omega} \lbrace\rho_{\sigma,+,0},\rho_{\sigma, -,0} \rbrace > 0,
\end{equation}
be the strict positive lower bound of initial densities.

In the case of $ \Omega = \mathbb T^2 \times (0,1) $, we impose the impermeable and complete slip boundary conditions:
\begin{equation}\label{bc:flat}
	\vec \tau \cdot \mathbb{S}_{\sigma,\pm}\vec{n}\big|_{z=0,1} = 0, \quad \vec u_{\sigma,\pm} \cdot \vec n \big|_{z=0,1} = 0,
\end{equation}
where
\begin{equation*}
	\vec \tau \in \left\lbrace
	\left( \begin{array}{c}
	1\\0\\0
	\end{array}\right), \quad 
	\left( \begin{array}{c}
	0\\1\\0
	\end{array}\right)
	\right\rbrace, \quad 
	\vec n = \left( \begin{array}{c}
	0\\0\\1
	\end{array}\right).
\end{equation*}
Equivalently, \eqref{bc:flat} can be written as
\begin{equation*}\tag{\ref{bc:flat}}
	\partial_z \vec v_{\sigma,\pm}\big|_{z=0,1} = 0, \quad w_{\sigma,\pm}\big|_{z=0,1} = 0,
\end{equation*}
where $ \vec v_{\sigma,\pm} $ and $ w_{\sigma,\pm} $ are the horizontal and vertical components of $ \vec u_{\sigma,\pm} $, respectively, i.e., $ \vec u_{\sigma,\pm}:= (\vec v_{\sigma,\pm},w_{\sigma,\pm})^\top $.

Notice, System \eqref{eq:2CMPNS} with either of the above boundary conditions, admits the following conservation/balance laws:
\begin{align}
	\text{Conservation of mass:} & \quad \dfrac{d}{dt} \int \rho_{\sigma,\pm} \idx = 0, \label{cnsvt:mass} \\
	\text{Conservation of total momentum:} & \quad \dfrac{d}{dt} \int \sum_{\pmindex}\rho_{\sigma, \mrm j} \vec u_{\sigma, \mrm j} \idx = 0, \label{cnsvt:momentum} \\
	\text{Balance of total energy:} & \quad \dfrac{d}{dt} E_\sigma + D_\sigma \nonumber \\
	&  + \sigma \int \rho_{\sigma, +} \rho_{\sigma,-}|\vec u_{\sigma,+} - \vec u_{\sigma,-}|^2\idx =0 \label{cnsvt:energy} ,
\end{align}
where the energy and the dissipation are defined as
\begin{align}
	E_\sigma & := \int \sum_{\pmindex}\biggl( \dfrac{1}{2} \rho_{\sigma,\mrm j} |\vec u_{\sigma,\mrm j}|^2 + \dfrac{R_\mrm j}{\gamma_\mrm j - 1} \rho_{\sigma, \mrm j}^\gamma  \biggr) \idx, \\
	D_\sigma & := \int \sum_{\pmindex}\biggl( \dfrac{\mu_\mrm j}{2} |\nabla \vec u_{\sigma,\mrm j} + \nabla^\top \vec u_{\sigma,\mrm j}|^2 + \lambda_\mrm j |\dv \vec u_{\sigma,\mrm j}|^2 \biggr)\idx.
\end{align}

Now we state our first main theorem in this paper:
\begin{thm}\label{thm:1}
	{\par\noindent\bf Case 1.} In the case of $ \Omega = \mathbb T^3 $, assume that
	\begin{equation}\label{asm:initial}
			\underline{\rho} > 0, \quad \rho_{\sigma,\pm,0} \in H^2(\Omega), \quad \vec u_{\sigma,\pm,0}\in H^2(\Omega).
	\end{equation}
	There exists $ T \in (0,\infty) $, depending only on the initial data and independent of $ \sigma $, such that $ \inf_{\Omega, 0\leq t \leq T}\rho_{\sigma,\pm} \geq \underline{\rho}/2 $ and
	\begin{equation}\label{fnt-space}
		\begin{gathered}
			\rho_{\sigma,\pm} \in L^\infty(0,T;H^2(\Omega)), \quad \dt \rho_{\sigma,\pm} \in L^\infty(0,T;L^2(\Omega)), \\
			\vec u_{\sigma,\pm} \in L^\infty(0,T;H^2(\Omega))\cap L^2(0,T;H^3(\Omega)), \\
			\dt \vec u_{\sigma,\pm} \in L^\infty(0,T;L^2(\Omega)) \cap L^2(0,T;H^1(\Omega)).
		\end{gathered}
	\end{equation}
	Moreover, it holds
	\begin{equation}\label{thm:est}
		\begin{aligned}
		& \sup_{0\leq s\leq T} \biggl\lbrace \norm{\rho_{\sigma,\pm}(s)}{\Hnorm{2}}^2 + \norm{\dt \rho_{\sigma,\pm}(s)}{\Lnorm{2}}^2 + 
		\norm{\vec u_{\sigma,\pm}(s)}{\Hnorm{2}}^2 \\
		& \qquad
		 + \norm{\dt \vec u_{\sigma,\pm}(s)}{\Lnorm{2}}^2 
		 + \sigma^2 \norm{(\vec u_{\sigma,+} - \vec u_{\sigma,-})(s)}{\Lnorm{2}}^2 \\
		 & \qquad
		 + \sigma \norm{(\vec u_{\sigma,+} - \vec u_{\sigma,-})(s)}{\Hnorm{1}}^2 \biggr\rbrace + \sigma^2 \int_0^T \norm{(\vec u_{\sigma,+} - \vec u_{\sigma,-})(s)}{\Hnorm{1}}^2 \,ds\\
		& \quad + \int_0^T \biggl\lbrace \norm{\vec u_{\sigma,\pm}(s)}{\Hnorm{3}}^2 + \norm{\dt \vec u_{\sigma,\pm}(s)}{\Hnorm{1}}^2 \biggr\rbrace \,ds \\
		& \quad + \sigma \int_0^T \biggl\lbrace \norm{(\vec u_{\sigma,+}-\vec u_{\sigma,-})(s)}{\Hnorm{2}}^2 +  \norm{\dt (\vec u_{\sigma,+}-\vec u_{\sigma,-})(s)}{\Lnorm{2}}^2 \biggr\rbrace \,ds\\
		& \qquad \leq C_{\mrm{in},T},
		\end{aligned}
	\end{equation}
	where $ C_{\mrm{in},T} \in (0,\infty) $ is a constant depending only on the initial data and $ T $, but independent of $ \sigma $. 
	
	{\par\noindent\bf Case 2.} In the case of $ \Omega = \mathbb T^2 \times (0,1) $, the same conclusions as in Case 1 hold, provided that: 
	\begin{itemize}
		\item \eqref{asm:initial} holds;
		\item In addition to \eqref{asm:general}, 
			\begin{equation}\label{asm:flat-1}
				\mu_\pm = \mu, \quad \lambda_\pm = \lambda;
			\end{equation}
		\item $ \vec u_{\sigma,\pm,0} $ satisfies \eqref{bc:flat};
		\item  $ \sigma \geq \sigma^* $, for some $ \sigma^* $ depending only on the initial data. 
	\end{itemize}
\end{thm}
\begin{proof}[Proof of Theorem \ref{thm:1}] The local well-posedness of strong solutions to \eqref{eq:2CMPNS}, where the life span may depend on $ \sigma $, follows straightforward from standard fixed point arguments. Through a continuity argument, with the uniform estimates established in sections \ref{sec:t3} and \ref{sec:flat}, below, respectively, one can extend the life span of strong solutions to a finite time, which is strictly positive, and uniform in $ \sigma $. In particular, \eqref{thm:est} follows from \eqref{est:rho-total-1}, \eqref{est:u-total-1}, \eqref{est:u-totla-1-02}, and \eqref{est:u-total-2}. This finishes the proof. 
\end{proof}
The goal of this paper is to study the asymptote of system \eqref{eq:2CMPNS} as 
\begin{equation}\label{asymptote}
	\sigma \rightarrow \infty,
\end{equation}
which is stated in the following:
\begin{thm}\label{thm-2}
	Under the conditions in Theorem \ref{thm:1}, in either case, there exists $ (\rho_\pm,\vec u) $, as the limit of $ (\rho_{\sigma,\pm},\vec u_{\sigma,\pm} ) $ as $ \sigma \rightarrow \infty $ in the sense of \eqref{limit:general-1}, below, such that it solves the two-phase fluid model \eqref{eq:limiting-eq} in $ \mathbb T^3 $ or $ \mathbb T^2 \times(0,1) $ with boundary conditions \eqref{bc:flat}, respectively. 
	
	In addition, \eqref{eq:conv-3} holds with $ \mathfrak S $ given by
	\begin{equation}\label{def:interaction}
	\begin{aligned}
	& \mathfrak S :=	\dfrac{1}{\rho_+ + \rho_-} \biggl\lbrack \dfrac{R_+ \gamma_+ }{\gamma_+ - 1} \nabla \rho_+^{\gamma_+ - 1} - \dfrac{R_- \gamma_- }{\gamma_- - 1} \nabla \rho_-^{\gamma_- - 1}  \\
	& \qquad + \biggl( \dfrac{\mu_+}{\rho_+} - \dfrac{\mu_-}{\rho_-}\biggr)\dv (\nabla \vec u + \nabla^\top \vec u) + \biggl( \dfrac{\lambda_+}{\rho_+} - \dfrac{\lambda_-}{\rho_-}\biggr) \nabla \dv \vec u \biggr\rbrack.
	\end{aligned}
	\end{equation}
\end{thm}
\begin{proof}[Proof of Theorem \ref{thm-2}] The arguments in section \ref{sec:limit} directly yield the theorem. 
\end{proof}

\begin{remark}
	The assumption \eqref{asm:initial} implies that
	\begin{gather*}
		\norm{\nabla \log \rho_{\sigma,\pm,0}}{\Lnorm{3}},~ \norm{\rho_{\sigma,\pm,0}^{1/2}}{\Hnorm{2}}, ~ \norm{\dt \rho_{\sigma,\pm}\big|_{t = 0}}{\Lnorm{2}}, \\
			\norm{\vec u_{\sigma,\pm,0}}{\Hnorm{2}}, ~ \norm{\dt \vec u_{\sigma,\pm}\big|_{t=0}}{\Lnorm{2}},
	\end{gather*}
	are bounded, which will be used in sections \ref{subsec:density-t3}, \ref{subsec:uniform-1}, and \ref{subsec:uniform-2}.
\end{remark}

\begin{remark}
	The method in this paper can be generated to more general domain $ \Omega \in \mathbb R^3 $, for instance, any bounded domain with smooth boundary, by applying standard localizing-in-space arguments. Also, with minor modifications, instead of \eqref{bc:flat}, one can also show similar results with no-slip boundary conditions, i.e., $ \vec u_{\sigma,\pm}\big|_{\partial \Omega} = 0$.
\end{remark}

\section{Uniform estimates: $ \Omega = \mathbb T^3 $}\label{sec:t3}

In this section, we study \eqref{eq:2CMPNS} in periodic domain $ \mathbb T^3 $. The goal is to obtain uniform in $ \sigma $ estimates in order to pass the limit \eqref{asymptote}.

\subsection{Temporal derivative estimates}

After applying $ \dt $ to system \eqref{eq:2CMPNS}, we end up with:
\begin{equation}\label{eq:2CMPNS-dt}
	\begin{cases}
		\partial_{tt}\rho_{\sigma, \pm} + \dv(\dt \rho_{\sigma,\pm} \vec u_{\sigma,\pm} ) + \dv ( \rho_{\sigma,\pm} \dt \vec u_{\sigma,\pm} ) = 0,\\
		\rho_{\sigma, \pm} \partial_{tt} \vec u_{\sigma, \pm} + \rho_{\sigma,\pm} \vec u_{\sigma, \pm} \cdot \nabla \dt \vec u_{\sigma, \pm} + \dt \rho_{\sigma, \pm} \partial_{t} \vec u_{\sigma, \pm} \\
		\qquad + \dt (\rho_{\sigma,\pm} \vec u_{\sigma, \pm}) \cdot \nabla \vec u_{\sigma, \pm} + \nabla \dt p_{\sigma, \pm} = \dv \dt \mathbb S_{\sigma, \pm} \\
		\qquad + \sigma \dt(\rho_{\sigma, \pm}\rho_{\sigma, \mp}) (\vec u_{\sigma,\mp} - \vec u_{\sigma, \pm}) + \sigma \rho_{\sigma, \pm}\rho_{\sigma, \mp} \dt (\vec u_{\sigma,\mp} - \vec u_{\sigma, \pm}).
	\end{cases} 
\end{equation}

Taking the $ L^2 $-inner product of \subeqref{eq:2CMPNS-dt}{2} with $ 2 \dt \vec u_{\sigma, \pm} $, after applying integration by parts and summing up the $ + $--estimate with the $ - $--estimate, leads to
\begin{equation}\label{est:001}
	\begin{aligned}
		& \dfrac{d}{dt} \sum_{\pmindex}\norm{\rho_{\sigma,\mrm j}^{1/2}\dt \vec u_{\sigma,\mrm j}}{\Lnorm{2}}^2 \\
		& \qquad + \sum_\pmindex \biggl\lbrace \mu_\mrm j \norm{\nabla\dt \vec u_{\sigma,\mrm j} + \nabla^\top \dt \vec u_{\sigma,\mrm j} }{\Lnorm{2}}^2 + 2 \lambda_\mrm j \norm{\dv \dt \vec u_{\sigma,\mrm j}}{\Lnorm{2}}^2 \biggr\rbrace \\
		& \qquad + 2 \sigma \int \rho_{\sigma,+}\rho_{\sigma,-} |\dt (\vec u_{\sigma, +} - \vec u_{\sigma, -})|^2 \idx \\
		& \quad = \sum_{j = 1}^{4}I_{j},
	\end{aligned}
\end{equation}
where
\begin{align*}
	I_{1} := & - 2 \sum_{\pmindex}\int \dt \rho_{\sigma, \mrm j} |\dt \vec u_{\sigma,\mrm j}|^2 \idx  , \\
	I_{2} := & - 2 \sum_{\pmindex} \int \bigl(\dt(\rho_{\sigma,\mrm j}\vec u_{\sigma,\mrm j}) \cdot \nabla \bigr) \vec u_{\sigma, \mrm j}  \cdot \dt \vec u_{\sigma,\mrm j} \idx , \\
	I_{3} := & 2 \sum_{\pmindex} R_\mrm j \gamma_\mrm j \int \rho_{\sigma,\mrm j}^{\gamma_\mrm{j} - 1} \dt \rho_{\sigma, \mrm j} \dv \dt \vec u_{\sigma,\mrm j} \idx, \\
	I_{4} := & 2 \sigma \int \dt(\rho_{\sigma,+}\rho_{\sigma,-}) ( \vec u_{\sigma,-} - \vec u_{\sigma,+}) \cdot \dt (\vec u_{\sigma,+} - \vec u_{\sigma, -}) \idx.
\end{align*}
Applying H\"older's inequality, one can obtain that
\begin{align*}
	& I_4 \leq \sigma \int \rho_{\sigma,+} \rho_{\sigma, -} |\dt(\vec u_{\sigma,+} - \vec u_{\sigma,-})|^2 \idx \\
	& \qquad + 4\sigma \norm{\dt ( \rho_{\sigma,+}\rho_{\sigma, -})^{1/2} }{\Lnorm{2}}^2 \norm{\vec u_{\sigma, +} - \vec u_{\sigma,-} }{\Lnorm{\infty}}^2.
\end{align*}
The estimates of $ I_1, \cdots, I_3 $ are standard, which we omit the details. Thus \eqref{est:001} yields, 
\begin{equation}\label{est:002}
	 \begin{aligned}	
	 	& \dfrac{d}{dt} \sum_{\pmindex}\norm{\rho_{\sigma,\mrm j}^{1/2}\dt \vec u_{\sigma,\mrm j}}{\Lnorm{2}}^2 \\
	 	& \qquad + \sum_\pmindex \biggl\lbrace \dfrac{\mu_\mrm j}{2} \norm{\nabla\dt \vec u_{\sigma,\mrm j} + \nabla^\top \dt \vec u_{\sigma,\mrm j} }{\Lnorm{2}}^2 +  \lambda_\mrm j \norm{\dv \dt \vec u_{\sigma,\mrm j}}{\Lnorm{2}}^2 \biggr\rbrace \\
	 	& \qquad +  \sigma \int \rho_{\sigma,+}\rho_{\sigma,-} |\dt (\vec u_{\sigma, +} - \vec u_{\sigma, -})|^2 \idx\\
	 	& \quad \leq 4 \sigma \norm{\dt ( \rho_{\sigma,+}\rho_{\sigma, -})^{1/2} }{\Lnorm{2}}^2 \norm{\vec u_{\sigma, +} - \vec u_{\sigma,-} }{\Lnorm{\infty}}^2 \\
	 	& \qquad + \mathcal H( \norm{\rho_{\sigma,\pm}}{\Lnorm{\infty}}, \norm{\nabla \rho_{\sigma,\pm}^{1/2}}{\Lnorm{6}},\norm{\vec u_{\sigma,\pm}}{\Lnorm{\infty}}, \norm{\nabla \vec u_{\sigma, \pm}}{\Lnorm{6}} , \\
	 	& \qquad \qquad \qquad \norm{\rho_{\sigma,\pm}^{1/2} \dt \vec u_{\sigma,\pm}}{\Lnorm{2}}),
	 \end{aligned}
\end{equation}
where we have used \subeqref{eq:2CMPNS}{1} to substitute $ \dt \rho_{\sigma,\pm} $, and we have applied the fact that, for any vector field $ \vec v: \Omega = \mathbb T^3 \mapsto \mathbb R^3 $, one has
\begin{equation}\label{id:korn-t3}
	2\norm{\nabla \vec v}{\Lnorm{2}}^2 + 2 \norm{\dv \vec v }{\Lnorm{2}}^2 = \norm{\nabla \vec v + \nabla^\top  \vec v}{\Lnorm{2}}^2.
\end{equation}

\subsection{Spatial derivative estimates}

After applying $ \partial^2 = \partial \partial $ to system \eqref{eq:2CMPNS}, with $ \partial \in \lbrace \partial_{x}, \partial_y,\partial_z \rbrace $, one can write down:
\begin{equation}\label{eq:2CMPNS-dd}
	\begin{cases}
	\dt \partial^2 \rho_{\sigma, \pm} + \dv (\partial^2 \rho_{\sigma,\pm} \vec u_{\sigma,\pm}) + 2 \dv (\partial \rho_{\sigma, \pm} \partial \vec u_{\sigma,\pm} ) \\
	\qquad + \dv ( \rho_{\sigma, \pm} \partial^2 \vec u_{\sigma,\pm} ) = 0, \\
	\rho_{\sigma,\pm} \dt \partial^2 \vec u_{\sigma,\pm} + \rho_{\sigma, \pm} \vec u_{\sigma,\pm} \cdot \nabla \partial^2 \vec u_{\sigma, \pm} + 2 \partial \rho_{\sigma, \pm} \dt \partial \vec u_{\sigma,\pm} \\
	\qquad + 2 \partial(\rho_{\sigma,\pm} \vec u_{\sigma,\pm}) \cdot \nabla \partial \vec u_{\sigma,\pm} + \partial^2 \rho_{\sigma, \pm} \dt \vec u_{\sigma,\pm} + \partial^2 (\rho_{\sigma, \pm}\vec u_{\sigma,\pm}) \cdot \nabla \vec u_{\sigma,\pm} \\
	\qquad + \nabla \partial^2 p_{\sigma,\pm} = \dv \partial^2 \mathbb S_{\sigma,\pm} + \sigma \rho_{\sigma,\pm}\rho_{\sigma,\mp} \partial^2 (\vec u_{\sigma,\mp} - \vec u_{\sigma,\pm}) \\
	\qquad + 2 \sigma \partial (\rho_{\sigma,\pm}\rho_{\sigma,\mp}) \partial (\vec u_{\sigma,\mp} - \vec u_{\sigma,\pm}) + \sigma \partial^2(\rho_{\sigma,\pm}\rho_{\sigma,\mp}) (\vec u_{\sigma,\mp} - \vec u_{\sigma,\pm}).
	\end{cases}
\end{equation}

\begin{remark}
	For any two functions $ f $ and $ g $, we have used
	$$ \partial^2 (fg) = \partial^2 f \times g + 2 \partial f \partial g + f \partial^2 g $$
	to represent
	$ \partial_{\mrm i}\partial_{\mrm j} (fg) = \partial_{\mrm i}\partial_{\mrm j} f \times g + \partial_{\mrm i} f \partial_{\mrm j} g + \partial_{\mrm j} f \partial_{\mrm i} g + f \times \partial_{\mrm i}\partial_{\mrm j} g $, $ \mrm i, \mrm j \in \lbrace x,y,z \rbrace $ for the sake of simplifying the notations.
\end{remark}

Taking the $ L^2 $-inner product of \subeqref{eq:2CMPNS-dd}{2} with $ 2 \partial^2 \vec u_{\sigma,\pm} $, after applying integration by parts and summing up the $ + $--estimate and the $-$--estimate, leads to
\begin{equation}\label{est:101}
	\begin{aligned}
	& \dfrac{d}{dt} \sum_{\pmindex} \norm{\rho_{\sigma,\mrm j}^{1/2} \partial^2 \vec u_{\sigma, \mrm j} }{\Lnorm{2}}^2 \\
	& \qquad + \sum_{\pmindex} \biggl\lbrace \mu_\mrm j \norm{\nabla \partial^2 \vec u_{\sigma,\mrm j} + \nabla^\top \partial^2 \vec u_{\sigma,\mrm j}}{\Lnorm{2}}^2 + 2 \lambda_\mrm j \norm{\dv \partial^2 \vec u_{\sigma,\mrm j}}{\Lnorm{2}}^2 \biggr\rbrace \\
	& \qquad + 2 \sigma \int \rho_{\sigma,+} \rho_{\sigma,-} |\partial^2 (\vec u_{\sigma,+}-\vec u_{\sigma,-})|^2 \idx\\
	& \quad = \sum_{j = 5}^{11} I_j,
	\end{aligned}
\end{equation}
where
\begin{align*}
	I_5 := & - 4 \sum_{\pmindex} \int \partial \rho_{\sigma,\mrm j} \dt \partial \vec u_{\sigma, \mrm j} \cdot \partial^2 \vec u_{\sigma,\mrm j} \idx , \\
	I_6 := & -4 \sum_\pmindex \int \bigl( \partial (\rho_{\sigma,\mrm j} \vec u_{\sigma,\mrm j}) \cdot \nabla \bigr) \partial \vec u_{\sigma,\mrm j} \cdot \partial^2 \vec u_{\sigma,\mrm j} \idx , \\
	I_7 := & - 2 \sum_{\pmindex} \int \partial^2 \rho_{\sigma,\mrm j} \dt \vec u_{\sigma, \mrm j} \cdot \partial^2 \vec u_{\sigma,\mrm j} \idx , \\
	I_8 := & - 2 \sum_\pmindex \int \bigl( \partial^2 (\rho_{\sigma,\mrm j} \vec u_{\sigma,\mrm j}) \cdot \nabla \bigr) \vec u_{\sigma,\mrm j} \cdot \partial^2 \vec u_{\sigma,\mrm j} \idx , \\
	I_9 := & 2 \sum_\pmindex R_\mrm j \int \partial^2 \rho_{\sigma,\mrm j}^{\gamma_\mrm j} \dv \partial^2 \vec u_{\sigma,\mrm j} \idx  , \\
	I_{10} := & 4 \sigma \int \partial(\rho_{\sigma,+}\rho_{\sigma, -}) \partial (\vec u_{\sigma, -} - \vec u_{\sigma,+}) \cdot \partial^2 (\vec u_{\sigma, +} - \vec u_{\sigma,-}) \idx , \\
	I_{11} := & 2 \sigma \int \partial^2(\rho_{\sigma,+}\rho_{\sigma, -}) (\vec u_{\sigma, -} - \vec u_{\sigma,+}) \cdot \partial^2 (\vec u_{\sigma, +} - \vec u_{\sigma,-}) \idx.
\end{align*}
Applying H\"older's inequality leads to 
\begin{align*}
	& I_{10} + I_{11} \leq \sigma \int \rho_{\sigma,+}\rho_{\sigma,-} |\partial^2 (\vec u_{\sigma,+} - \vec u_{\sigma,-})|^2 \idx \\
	& \qquad + \sigma \bigl( 2 \norm{\partial\log (\rho_{\sigma, +}\rho_{\sigma, -})}{\Lnorm{3}}^2 \norm{\partial(\rho_{\sigma, +}\rho_{\sigma, -})^{1/2}}{\Lnorm{6}}^2 \\
	& \qquad\qquad + 8 \norm{\partial^2(\rho_{\sigma, +}\rho_{\sigma, -})^{1/2}}{\Lnorm{2}}^2  \bigr) \times
	\norm{\vec u_{\sigma,+}-\vec u_{\sigma,-}}{\Lnorm{\infty}}^2\\
	& \qquad + 32 \sigma \norm{\partial(\rho_{\sigma, +}\rho_{\sigma, -})^{1/2}}{\Lnorm{6}}^2 \norm{\partial(\vec u_{\sigma,+} - \vec u_{\sigma,-})}{\Lnorm{3}}^2.
\end{align*}
The estimates of $ I_{5}, \cdots, I_{9} $ are standard, which we omit the details again. Thus \eqref{est:101} yields, for any $ \delta \in (0,1) $,
\begin{equation}\label{est:102}
	\begin{aligned}
	& \dfrac{d}{dt} \sum_{\pmindex} \norm{\rho_{\sigma,\mrm j}^{1/2} \partial^2 \vec u_{\sigma, \mrm j} }{\Lnorm{2}}^2 \\
	& \qquad + \sum_{\pmindex} \biggl\lbrace \dfrac{\mu_\mrm j}{2} \norm{\nabla \partial^2 \vec u_{\sigma,\mrm j} + \nabla^\top \partial^2 \vec u_{\sigma,\mrm j}}{\Lnorm{2}}^2 +  \lambda_\mrm j \norm{\dv \partial^2 \vec u_{\sigma,\mrm j}}{\Lnorm{2}}^2 \biggr\rbrace \\
	& \qquad +  \sigma \int \rho_{\sigma,+} \rho_{\sigma,-} |\partial^2 (\vec u_{\sigma,+}-\vec u_{\sigma,-})|^2 \idx\\
	& \quad \leq \sigma \bigl( 2 \norm{\partial\log (\rho_{\sigma, +}\rho_{\sigma, -})}{\Lnorm{3}}^2 \norm{\partial(\rho_{\sigma, +}\rho_{\sigma, -})^{1/2}}{\Lnorm{6}}^2 \\
	& \qquad\qquad + 8 \norm{\partial^2(\rho_{\sigma, +}\rho_{\sigma, -})^{1/2}}{\Lnorm{2}}^2  \bigr) \times
	\norm{\vec u_{\sigma,+}-\vec u_{\sigma,-}}{\Lnorm{\infty}}^2\\
	& \qquad + 32 \sigma \norm{\partial(\rho_{\sigma, +}\rho_{\sigma, -})^{1/2}}{\Lnorm{6}}^2 \norm{\partial(\vec u_{\sigma,+} - \vec u_{\sigma,-})}{\Lnorm{3}}^2 \\
	& \qquad + \delta \sum_{\pmindex} \norm{\nabla \dt \vec u_{\sigma,\mrm j}}{\Lnorm{2}}^2 + \mathcal H(C_\delta, \norm{\rho_{\sigma,\pm}}{\Lnorm{\infty}}, \norm{\nabla \rho_{\sigma,\pm}^{1/2}}{\Lnorm{6}}, \\
	& \qquad \qquad \qquad \norm{\nabla\log \rho_{\sigma,\pm}}{\Lnorm{3}}, \norm{\nabla^2\rho_{\sigma,\pm}^{1/2}}{\Lnorm{2}},
	\norm{\vec u_{\sigma,\pm}}{\Lnorm{\infty}},\\
	& \qquad \qquad\qquad  \norm{\nabla \vec u_{\sigma,\pm}}{\Lnorm{6}} , \norm{\rho_{\sigma,\pm}^{1/2}\nabla^2 \vec u_{\sigma,\pm}}{\Lnorm{2}}, \norm{\rho_{\sigma,\pm}^{1/2}\dt \vec u_{\sigma,\pm}}{\Lnorm{2}}
	).
	\end{aligned}
\end{equation}

Similar estimates also hold for $ \norm{\rho_{\sigma,\pm}^{1/2} \partial \vec u_{\sigma,\pm}}{\Lnorm{2}} $. We omit the detail and only record the result here:
\begin{equation}\label{est:103}
	\begin{aligned}
	& \dfrac{d}{dt}\sum_{\pmindex} \norm{\rho_{\sigma,\mrm j}^{1/2} \partial \vec u_{\sigma, \mrm j} }{\Lnorm{2}}^2 \\
	& \qquad + \sum_{\pmindex} \biggl\lbrace \dfrac{\mu_\mrm j}{2} \norm{\nabla \partial \vec u_{\sigma,\mrm j} + \nabla^\top \partial \vec u_{\sigma,\mrm j}}{\Lnorm{2}}^2 +  \lambda_\mrm j \norm{\dv \partial \vec u_{\sigma,\mrm j}}{\Lnorm{2}}^2 \biggr\rbrace \\
	& \qquad +  \sigma \int \rho_{\sigma,+} \rho_{\sigma,-} |\partial (\vec u_{\sigma,+}-\vec u_{\sigma,-})|^2 \idx\\
	& \quad \leq 4 \sigma \norm{\partial ( \rho_{\sigma,+}\rho_{\sigma, -})^{1/2} }{\Lnorm{2}}^2 \norm{\vec u_{\sigma, +} - \vec u_{\sigma,-} }{\Lnorm{\infty}}^2\\
	& \qquad + \mathcal H( \norm{\rho_{\sigma,\pm}}{\Lnorm{\infty}}, \norm{\nabla \rho_{\sigma,\pm}^{1/2}}{\Lnorm{6}},\norm{\vec u_{\sigma,\pm}}{\Lnorm{\infty}}, \norm{\nabla \vec u_{\sigma, \pm}}{\Lnorm{6}} , \\
	& \qquad \qquad \qquad \norm{\rho_{\sigma,\pm}^{1/2} \dt \vec u_{\sigma,\pm}}{\Lnorm{2}}).
	\end{aligned}
\end{equation}

\subsection{Estimates on the densities}\label{subsec:density-t3}

First, we derive the estimate of $ \norm{ \rho_{\sigma,\pm}^{1/2}}{\Hnorm{2}} $. Recall from \subeqref{eq:2CMPNS}{1}, one has
\begin{equation}\label{eq:rho-to-half}
	2 \dt \rho_{\sigma,\pm}^{1/2} + 2 \vec u_{\sigma,\pm} \cdot \nabla \rho_{\sigma, \pm}^{1/2} + \rho_{\sigma, \pm}^{1/2}\dv \vec u_{\sigma,\pm} = 0.
\end{equation}
Then performing standard $ H^s $-estimates yields
\begin{equation}\label{est:rho-h2}
	\norm{\rho_{\sigma,\pm}^{1/2}(t)}{\Hnorm{2}}\leq \norm{\rho_{\sigma,\pm,0}^{1/2}}{\Hnorm{2}} e^{C\int_0^t \norm{\nabla \vec u_{\sigma,\pm}(s)}{\Hnorm{2}} \,ds},
\end{equation}
and therefore
\begin{equation}\label{est:rho-embedding}
	\norm{\rho_{\sigma,\pm}}{\Lnorm{\infty}}, \norm{\nabla \rho_{\sigma,\pm}^{1/2}}{\Lnorm{6}}^2 \leq C \norm{\rho_{\sigma,\pm,0}^{1/2}}{\Hnorm{2}}^2 e^{C\int_0^t \norm{\nabla \vec u_{\sigma,\pm}(s)}{\Hnorm{2}} \,ds}.
\end{equation}


Next, we shall derive the estimate of $ \norm{\nabla\log\rho_{\sigma,\pm}}{\Lnorm{3}} $. 
$ \log \rho_{\sigma,\pm} $ satisfies
\begin{equation}\label{eq:rho-w-log}
	\dt \log \rho_{\sigma,\pm} + \vec u_{\sigma,\pm} \cdot \nabla \log \rho_{\sigma, \pm} + \dv \vec u_{\sigma,\pm} =0 .
\end{equation}
Therefore, one can derive
\begin{gather*}
	\dfrac{d}{dt} \norm{\nabla \log \rho_{\sigma,\pm}}{\Lnorm{3}} \leq 2 \norm{\nabla \vec u_{\sigma,\pm}}{\Lnorm{\infty}} \norm{\nabla \log \rho_{\sigma,\pm}}{\Lnorm{3}} + \norm{\nabla^2 \vec u_{\sigma,\pm}}{\Lnorm{3}}\\
	\leq C \norm{\nabla \vec u_{\sigma,\pm}}{\Hnorm{2}} \norm{\nabla \log \rho_{\sigma,\pm}}{\Lnorm{3}} + C \norm{\nabla \vec u_{\sigma,\pm}}{\Hnorm{2}} ,
\end{gather*}
which yields
\begin{equation}\label{est:rho-w-log-w13}
	\norm{\nabla \log \rho_{\sigma,\pm}(t)}{\Lnorm{3}} \leq (C + \norm{\nabla \log \rho_{\sigma,\pm,0}}{\Lnorm{3}}) e^{C\int_0^t \norm{\nabla \vec u_{\sigma,\pm}(s)}{\Hnorm{2}}} \,ds - C.
\end{equation}


In addition, applying $ \dt $ to \eqref{eq:rho-to-half} leads to
\begin{equation}\label{eq:rho-to-half-dt}
	\begin{gathered}
	2 \dt^2 \rho_{\sigma,\pm}^{1/2} + 2 \vec u_{\sigma,\pm} \cdot \nabla \dt \rho_{\sigma, \pm}^{1/2} + \dt \rho_{\sigma, \pm}^{1/2} \dv \vec u_{\sigma,\pm} \\
	+ 2 \dt \vec u_{\sigma,\pm} \cdot \nabla \rho_{\sigma, \pm}^{1/2} + \rho_{\sigma, \pm}^{1/2} \dv \dt \vec u_{\sigma,\pm} = 0.
	\end{gathered}
\end{equation}
Then after taking the $ L^2 $-inner product of \eqref{eq:rho-to-half-dt} with $ \dt \rho_{\sigma,\pm} $ and applying integration by parts in the resultant, one has
\begin{align*}
	& \dfrac{d}{dt} \norm{\dt \rho_{\sigma,\pm}^{1/2}}{\Lnorm{2}}^2 = - 2 \int( \dt \vec u_{\sigma,\pm} \cdot \nabla) \rho_{\sigma, \pm}^{1/2} \cdot \dt \rho_{\sigma, \pm}^{1/2} \idx \\
	& \qquad \qquad - \int \rho_{\sigma, \pm}^{1/2} \dv \dt \vec u_{\sigma,\pm} \cdot \dt \rho_{\sigma, \pm}^{1/2} \idx \\
	& \qquad \leq 2 \norm{\dt \vec u_{\sigma,\pm}}{\Lnorm{3}} \norm{\nabla \rho_{\sigma,\pm}^{1/2}}{\Lnorm{6}} \norm{\dt \rho_{\sigma,\pm}^{1/2}}{\Lnorm{2}} \\
	& \qquad \qquad + \norm{\rho_{\sigma,\pm}}{\Lnorm{\infty}}^{1/2} \norm{\dv \dt \vec u_{\sigma,\pm}}{\Lnorm{2}}\norm{\dt \rho_{\sigma,\pm}^{1/2}}{\Lnorm{2}},
\end{align*}
which implies
\begin{equation}\label{est:rho-dt}
	\begin{aligned}
		& \norm{\dt \rho_{\sigma,\pm}^{1/2}(t)}{\Lnorm{2}} \leq \norm{\dt \rho_{\sigma,\pm,0}^{1/2}}{\Lnorm{2}} \\
		& \qquad +  \sup_{0\leq s\leq t} \norm{\nabla \rho_{\sigma,\pm}^{1/2}(s)}{\Lnorm{6}} \times \int_0^t \norm{\dt \vec u_{\sigma,\pm}(s)}{\Lnorm{3}}\,ds \\
		& \qquad +  \sup_{0\leq s\leq t} \norm{\rho_{\sigma, \pm}(s)}{\Lnorm{\infty}}^{1/2} \times \int_0^t \norm{\dv \dt \vec u_{\sigma,\pm}(s)}{\Lnorm{2}}\,ds,
	\end{aligned}
\end{equation}
where $ \dt \rho_{\sigma,\pm,0}^{1/2} := \dt \rho_{\sigma,\pm}^{1/2}\big|_{t=0} $ is the initial data of $ \dt \rho_{\sigma,\pm}^{1/2} $ defined by \eqref{eq:rho-to-half}.


Last but not least, we will need to derive the lower bounds of $ \rho_{\sigma,\pm} $. From \subeqref{eq:2CMPNS}{1}, one can write down
\begin{align*}
	& \dt (\underline{\rho}- M(t) - \rho_{\sigma,\pm} ) + \vec u_{\sigma,\pm} \cdot \nabla (\underline{\rho}- M(t) - \rho_{\sigma,\pm} ) \\
	& \quad + (\underline{\rho}- M(t) - \rho_{\sigma,\pm} ) \dv \vec u_{\sigma,\pm} = - \dt M(t) + (\underline\rho - M(t)) \dv \vec u_{\sigma,\pm} \leq 0,
\end{align*}
where 
$$
	M(t) := \underline{\rho} \int_0^t \norm{\dv \vec u_{\sigma,\pm}(s)}{\Lnorm{\infty}} \,ds \times e^{\int_0^t \norm{\dv \vec u_{\sigma,\pm}(s)}{\Lnorm{\infty}} \,ds}.
$$
Therefore, testing the above equation with $ (\underline{\rho}- M(t) - \rho_{\sigma,\pm} )^+ $ leads to
\begin{equation*}
	\norm{(\underline{\rho}- M(t) - \rho_{\sigma,\pm} )^+}{\Lnorm{2}}^2 \leq 0,
\end{equation*}
and hence
\begin{equation}\label{est:rho-lower-bound}
\inf_{\Omega}\rho_{\sigma,\pm}(t) \geq \underline{\rho} (1 - \int_0^t \norm{\dv \vec u_{\sigma,\pm}(s)}{\Lnorm{\infty}} \,ds \times e^{\int_0^t \norm{\dv \vec u_{\sigma,\pm}(s)}{\Lnorm{\infty}} \,ds}).
\end{equation}

\subsection{Uniform-in-$ \sigma $ estimates}\label{subsec:uniform-1}

In this subsection, we summarize the estimates above and derive the uniform-in-$ \sigma $ estimates. Let $ T > 0 $ be the uniform-in-$\sigma$ life span of solutions to \eqref{eq:2CMPNS}, and denote by 
\begin{equation}\label{def:bound-dsptn-cauchy}
	\begin{gathered}
	\sup_{0\leq s\leq T}\norm{\dt \vec u_{\sigma,\pm}(s)}{\Lnorm{2}}^2+ \int_0^T \biggl\lbrace \norm{\nabla \vec u_{\sigma,\pm}(s)}{\Hnorm{2}}^2 + \norm{\nabla \dt \vec u_{\sigma,\pm}(s)}{\Lnorm{2}}^2 \biggr\rbrace \,ds \\
	 \leq \mathfrak M, \end{gathered}
\end{equation}
which will be shown to be independent of $ \sigma $. 

To shorten the notations, we will use $ C_\mrm{in}\in (0,\infty) $ through out this section to represent a generic constant depending only on
\begin{gather*}
	\underline{\rho} > 0, ~ \norm{\nabla \log \rho_{\sigma,\pm,0}}{\Lnorm{3}}, ~ \norm{\rho_{\sigma,\pm,0}^{1/2}}{\Hnorm{2}}, ~ \norm{\vec u_{\sigma,\pm,0}}{\Hnorm{2}}, ~ \norm{\vec u_{\sigma,\pm,1}}{\Lnorm{2}},
\end{gather*} 
but independent of $ \sigma $, which might be different from line to line, where $ \vec u_{\sigma, \pm, 1} = \dt \vec u_{\sigma,\pm}\big|_{t=0} $ are the initial data of $ \dt \vec u_{\sigma,\pm} $ defined by equation \subeqref{eq:2CMPNS}{2}.

Applying H\"older's inequality, one can obtain that, for $ T $ small enough and any $ t \in (0,T) $, 
\begin{equation}\label{chosing-time-001}
	\begin{gathered}
	\int_0^t \norm{\dt \vec u_{\sigma,\pm}(s)}{\Lnorm{3}} \,ds \leq C T \sup_{0\leq s\leq T}\norm{\dt \vec u_{\sigma,\pm}(s)}{\Lnorm{2}} \\
	 + C T^{1/2} \biggl(\int_0^T \norm{\nabla \dt \vec u_{\sigma,\pm}}{\Lnorm{2}}^2 \biggr)^{1/2} \leq C(T+T^{1/2})\mathfrak M^{1/2} \leq 1 ,\\
	\int_0^t \norm{\dv \dt \vec u_{\sigma,\pm}(s)}{\Lnorm{2}} \,ds \leq C T^{1/2}\biggl(\int_0^T \norm{\nabla \dt \vec u_{\sigma,\pm}}{\Lnorm{2}}^2 \biggr)^{1/2} \\
	 \leq CT^{1/2} \mathfrak M^{1/2} \leq 1  ,
	\\
	\int_0^t \norm{\nabla \vec u_{\sigma,\pm}(s)}{\Hnorm{2}} \,ds \leq T^{1/2} \biggl(\int_{0}^{T}\norm{\nabla \vec u_{\sigma,\pm}(s)}{\Hnorm{2}}^2 \,ds \biggr)^{1/2} \\
	 \leq T^{1/2}\mathfrak M^{1/2} \leq \mathfrak c_1,
	\end{gathered}
\end{equation}
with some uniform constant $ \mathfrak c_1 $, independent of $ \sigma $ and small enough such that
\begin{gather*}
	\int_0^t \norm{\dv \vec u_{\sigma,\pm}(s)}{\Lnorm{\infty}} \,ds \times e^{\int_0^t \norm{\dv \vec u_{\sigma,\pm}(s)}{\Lnorm{\infty}} \,ds}\\
	\leq C \int_0^t \norm{\nabla \vec u_{\sigma,\pm}(s)}{\Hnorm{2}} \,ds \times e^{C \int_0^t \norm{\nabla \vec u_{\sigma,\pm}(s)}{\Hnorm{2}} \,ds}
	\\
	\leq C \mathfrak c_1 \times e^{C \mathfrak c_1} \leq 1/2. 
\end{gather*}
Then from \eqref{est:rho-h2}, \eqref{est:rho-embedding},  \eqref{est:rho-w-log-w13},  \eqref{est:rho-dt},
and
\eqref{est:rho-lower-bound}, it follows that, for $ t \in (0,T] $,
\begin{equation}\label{est:rho-total-1}
\begin{gathered}
	\inf_{\Omega}\rho_{\sigma,\pm}(t) \geq \underline{\rho}/2,\\
	  \norm{\rho_{\sigma,\pm}^{1/2}(t)}{\Hnorm{2}}, ~ \norm{\rho_{\sigma,\pm}(t)}{\Lnorm{\infty}}, ~ \norm{\nabla \rho_{\sigma,\pm}^{1/2}(t)}{\Lnorm{6}}, \\
	~ \norm{\nabla\log \rho_{\sigma, \pm}(t)}{\Lnorm{3}},~ \norm{\dt \rho_{\sigma, \pm}^{1/2}(t)}{\Lnorm{2}} \leq C_{\mrm {in}}.
\end{gathered}
\end{equation}

Meanwhile, applying the Gagliardo--Nirenberg inequality, one has
\begin{equation}\label{ineq:interpolation-1}
\begin{gathered}
	\norm{\vec u_{\sigma,+}-\vec u_{\sigma,-}}{\Lnorm{\infty}}\leq C \norm{\vec u_{\sigma,+}-\vec u_{\sigma,-}}{\Hnorm{2}}^{3/4}\norm{\vec u_{\sigma,+}-\vec u_{\sigma,-}}{\Lnorm{2}}^{1/4}, \\
	\norm{\nabla (\vec u_{\sigma,+}-\vec u_{\sigma,-})}{\Lnorm{3}}\leq C \norm{\vec u_{\sigma,+}-\vec u_{\sigma,-}}{\Hnorm{2}}^{3/4}\norm{\vec u_{\sigma,+}-\vec u_{\sigma,-}}{\Lnorm{2}}^{1/4}.
\end{gathered}
\end{equation}

Therefore, after collecting \eqref{est:002}, \eqref{est:102}, and \eqref{est:103}, with suitably small $ \delta \in (0,1) $, substituting inequalities \eqref{est:rho-total-1} and \eqref{ineq:interpolation-1}, and applying \eqref{id:korn-t3}, Sobolev embedding inequalities, and Young's inequality, one can obtain
\begin{equation}\label{energy-est:total-t2}
	\begin{aligned}
	& \dfrac{d}{dt} E_{\sigma,1} + D_{\sigma,1} + \sigma \int \rho_{\sigma,+} \rho_{\sigma,-} \biggl(|\dt (\vec u_{\sigma,+}-\vec u_{\sigma,-})|^2 + |\nabla (\vec u_{\sigma,+}-\vec u_{\sigma,-})|^2 \\
	& \qquad \qquad \qquad + |\nabla^2 (\vec u_{\sigma,+}-\vec u_{\sigma,-})|^2   \biggr) \idx \\
	& \quad \leq \mathfrak c_2 \sigma \int \rho_{\sigma,+} \rho_{\sigma,-} |\vec u_{\sigma,+}-\vec u_{\sigma,-}|^2 \idx + \mathcal H( E_\sigma, E_{\sigma,1}, C_\mrm{in} ),
	\end{aligned}
\end{equation}
where $ \mathfrak c_2 \in (0,\infty) $ is a constant depending on $ C_\mrm{in} $, but independent of $ \sigma $, and
\begin{align}
	& \begin{aligned} & E_{\sigma,1} := \sum_{\pmindex} \biggl\lbrace \norm{\rho_{\sigma,\mrm j}^{1/2} \dt \vec u_{\sigma,\mrm j}}{\Lnorm{2}}^2 + \norm{\rho_{\sigma,\mrm j}^{1/2} \nabla \vec u_{\sigma,\mrm j}}{\Lnorm{2}}^2 \\
	& \qquad \qquad + \norm{\rho_{\sigma,\mrm j}^{1/2} \nabla^2 \vec u_{\sigma,\mrm j}}{\Lnorm{2}}^2 \biggr\rbrace, \end{aligned} \\
	& \begin{aligned} & D_{\sigma,1} : = \sum_{\pmindex} \bigg\lbrace \mu_\mrm j \norm{\nabla \dt \vec u_{\sigma,\mrm j}}{\Lnorm{2}}^2 +  (\lambda_\mrm j + \mu_+\mrm j) \norm{\dv \dt \vec u_{\sigma,\mrm j}}{\Lnorm{2}}^2 \\ 
	& \qquad \qquad + \mu_\mrm j \norm{\nabla^2\vec u_{\sigma,\mrm j} }{\Lnorm{2}}^2 +  (\lambda_\mrm j + \mu_\mrm j) \norm{\nabla \dv \vec u_{\sigma,\mrm j}}{\Lnorm{2}}^2 \\ 
	& \qquad \qquad + \mu_\mrm j \norm{\nabla^3 \vec u_{\sigma,\mrm j}}{\Lnorm{2}}^2 +  (\lambda_\mrm j + \mu_\mrm j) \norm{\nabla^2 \dv \vec u_{\sigma,\mrm j}}{\Lnorm{2}}^2
	\biggr\rbrace.
	\end{aligned}
\end{align}
Consequently, \eqref{cnsvt:energy} and \eqref{energy-est:total-t2} yield that, there is some $ T^* \in (0,T] $, independent of $ \sigma $, such that for any $ t \in (0, T^*) $,
\begin{equation}\label{est:total-final-t2}
	\begin{aligned}
		& \sup_{0\leq s\leq t} \bigl\lbrace E_\sigma(s) + E_{\sigma,1}(s) \bigr\rbrace + \int_0^t \bigl\lbrace D_\sigma(s) + D_{\sigma,1}(s) \bigr\rbrace \,ds \\
		& \qquad + \sigma \int_0^t \bigl\lbrace \norm{(\vec u_{\sigma,+}-\vec u_{\sigma,-})(s)}{\Hnorm{2}}^2 +  \norm{\dt (\vec u_{\sigma,+}-\vec u_{\sigma,-})(s)}{\Lnorm{2}}^2 \bigr\rbrace \,ds\\
		& \quad\leq C_\mrm{in},
	\end{aligned}
\end{equation}
which, in particular, implies
\begin{equation}\label{est:update-bound}
	\begin{gathered}
	\sup_{0\leq s\leq T^*}\norm{\dt \vec u_{\sigma,\pm}(s)}{\Lnorm{2}}^2+ \int_0^{T^*} \biggl\lbrace \norm{\nabla \vec u_{\sigma,\pm}(s)}{\Hnorm{2}}^2 + \norm{\nabla \dt \vec u_{\sigma,\pm}(s)}{\Lnorm{2}}^2 \biggr\rbrace \,ds \\ \leq \mathfrak M^*(C_{\mrm{in}}).\end{gathered}
\end{equation}

Now, we update $ \mathfrak M $ to $ \mathfrak M^*(C_\mrm{in}) $, and $ T $ to $ T^* $ accordingly. To this end, we conclude this section with \eqref{est:rho-total-1} and the following uniform estimates:
\begin{equation}\label{est:u-total-1}
	\begin{aligned}
	& \sup_{0\leq s\leq T} \biggl\lbrace \norm{\vec u_{\sigma,\pm}(s)}{\Hnorm{2}}^2 + \norm{\dt \vec u_{\sigma,\pm}(s)}{\Lnorm{2}}^2 \biggr\rbrace \\
	& \quad + \int_0^T \biggl\lbrace \norm{\vec u_{\sigma,\pm}(s)}{\Hnorm{3}}^2 + \norm{\dt \vec u_{\sigma,\pm}(s)}{\Hnorm{1}}^2 \biggr\rbrace \,ds \\
	& \quad + \sigma \int_0^T \biggl\lbrace \norm{(\vec u_{\sigma,+}-\vec u_{\sigma,-})(s)}{\Hnorm{2}}^2 +  \norm{\dt (\vec u_{\sigma,+}-\vec u_{\sigma,-})(s)}{\Lnorm{2}}^2 \biggr\rbrace \,ds\\
	& \qquad \leq C_{\mrm{in},T}.
	\end{aligned}
\end{equation}

On the other hand, we can rewrite \subeqref{eq:2CMPNS}{2} as
\begin{equation}\label{eq:mmt-reqrite}
\dt \vec u_{\sigma,\pm} + \vec u_{\sigma,\pm} \cdot \nabla \vec u_{\sigma,\pm} + \dfrac{\nabla p_{\sigma,\pm} - \dv \mathbb S_{\sigma,\pm}}{\rho_{\sigma,\pm}} = \sigma \rho_{\sigma,\mp} (\vec u_{\sigma,\mp} - \vec u_{\sigma,\pm}).
\end{equation}
Then after subtracting the $ + $--equation by the $ - $--equation of \eqref{eq:conv-2}, one obtains
\begin{equation}\label{eq:+--}
\begin{aligned}
& \dt (\vec u_{\sigma,+} - \vec u_{\sigma,-}) + \vec u_{\sigma,+} \cdot \nabla \vec u_{\sigma,+} - \vec u_{\sigma,-} \cdot \nabla \vec u_{\sigma,-} \\
& \qquad + \dfrac{\nabla p_{\sigma,+} - \dv \mathbb S_{\sigma,+}}{\rho_{\sigma,+}} - \dfrac{\nabla p_{\sigma,-} - \dv \mathbb S_{\sigma,-}}{\rho_{\sigma,-}}\\
& \quad = \sigma ( \rho_{\sigma, +} + \rho_{\sigma, -}) ( \vec u_{\sigma,-} - \vec u_{\sigma,+}).
\end{aligned}
\end{equation}
Therefore, utilizing \eqref{est:rho-total-1} and \eqref{est:u-total-1}, one can derive from \eqref{eq:+--} that
\begin{equation}\label{est:u-totla-1-02}
	\begin{aligned}
	& \sup_{0\leq s\leq T}\biggl\lbrace \sigma^2 \norm{(\vec{u}_{\sigma,+} - \vec  u_{\sigma,-})(s)}{\Lnorm{2}}^2 + \sigma \norm{(\vec{u}_{\sigma,+} - \vec  u_{\sigma,-})(s)}{\Hnorm{1}}^2 \biggr\rbrace \\
	&\qquad + \sigma^2 \int_0^T  \norm{(\vec{u}_{\sigma,+} - \vec  u_{\sigma,-})(s)}{\Hnorm{1}}^2 \leq C_{\mrm{in},T},
	\end{aligned}
\end{equation}
where we have applied Cauchy's inequality and the fact that
\begin{equation*}
	\norm{\vec{u}_{\sigma,+} - \vec  u_{\sigma,-}}{\Hnorm{1}}^2 \leq \norm{\vec{u}_{\sigma,+} - \vec  u_{\sigma,-}}{\Lnorm{2}}\norm{\vec{u}_{\sigma,+} - \vec  u_{\sigma,-}}{\Hnorm{2}}.
\end{equation*}


\section{Uniform estimates: $ \Omega = \mathbb T^2 \times (0,1) $}\label{sec:flat}

In this section, we study \eqref{eq:2CMPNS} in domain $ \mathbb T^2 \times (0,1) $. The goal is to obtain uniform estimates in order to pass the limit \eqref{asymptote}.
\subsection{Known estimates}

First, we shall record the estimates from section \ref{sec:t3} that still work, with minor modifications. In order to do so, we remind readers that, instead of using \eqref{id:korn-t3}, we will replace it by Korn's inequality in this section. 

Then estimate \eqref{est:002} holds, and all estimates in section \ref{subsec:density-t3} hold. 

Next, directly, one can check
\begin{equation}\label{est:001-flat}
	\dfrac{d}{dt} \norm{\nabla \vec u_{\sigma,\pm}}{\Lnorm{2}}^2 \leq 2 \norm{\nabla \vec u_{\sigma,\pm}}{\Lnorm{2}}\norm{\nabla\dt \vec u_{\sigma,\pm}}{\Lnorm{2}}.
\end{equation}


Thus, we only need to obtain the uniform $ H^2 $-- and $ H^3 $--estimates of $ \vec u_{\sigma,\pm} $ to close the uniform estimates. This will be done in the next two subsections.

\subsection{Normal derivative estimates: $ H^2 $}

We start by rewriting
\begin{equation}\label{id:viscosity-tensor}
	\begin{aligned}
	& \dv \mathbb S_{\sigma,\pm}  = \mu \Delta \vec u_{\sigma,\pm} + (\mu + \lambda) \nabla \dv \vec u_{\sigma,\pm}\\
	& \quad = \mu \Delta_h \vec u_{\sigma,\pm} + (\mu+\lambda) \left(\begin{array}{c}
	\nabla_h \dv \vec u_{\sigma,\pm} \\
	\partial_z \dvh \vec v_{\sigma,\pm}
	\end{array} \right) + \left(\begin{array}{c}
	\mu \partial_{zz} \vec v_{\sigma,\pm} \\
	(2\mu+\lambda) \partial_{zz} w_{\sigma,\pm}
	\end{array}\right),
	\end{aligned}
\end{equation}
and denote by 
\begin{equation}\label{def:lhs-mmt}
	\mathcal A_{\sigma,\pm} := \rho_{\sigma,\pm} \dt \vec u_{\sigma,\pm} + \rho_{\sigma,\pm} \vec u_{\sigma,\pm} \cdot \nabla \vec u_{\sigma,\pm} + \nabla p_{\sigma,\pm}.
\end{equation}
Thus \subeqref{eq:2CMPNS}{2} can be written as
	\begin{equation*}
		\mu \Delta \vec u_{\sigma,\pm} + (\mu + \lambda) \nabla \dv \vec u_{\sigma,\pm} + \sigma \rho_{\sigma,\pm} \rho_{\sigma,\mp} (\vec u_{\sigma,\mp} - \vec u_{\sigma,\pm}) = \mathcal A_{\sigma,\pm}.
	\end{equation*}
Therefore, after subtracting and adding the $+$--equation and $-$--equation, respectively, one has
\begin{gather}
	\mu\Delta \vec u_{\sigma,++-} +(\mu+\lambda) \nabla\dv \vec u_{\sigma,++-} = \mathcal A_{\sigma,+} + \mathcal A_{\sigma,-}, \label{eq:flat++-}  \\
	\mu\Delta \vec u_{\sigma,+--}  +(\mu+\lambda) \nabla\dv \vec u_{\sigma,+--} - 2 \sigma \rho_{\sigma,+}\rho_{\sigma,-} \vec u_{\sigma,+--}  = \mathcal A_{\sigma,+} - \mathcal A_{\sigma,-}, \label{eq:flat+--}
\end{gather}
where
\begin{equation}\label{def:flat-velocity}
	\vec u_{\sigma,++-}:= \vec u_{\sigma,+} + \vec u_{\sigma,-}, \quad \vec u_{\sigma,+--}:=\vec u_{\sigma,+} - \vec u_{\sigma,-}.
\end{equation}

After testing \eqref{eq:flat++-} with $ - \vec u_{\sigma,++-} $, $ \Delta_h \vec u_{\sigma,++-} $, respectively, and applying integration by parts and H\"older's inequality in the resultant, one can derive
\begin{equation}\label{est:flat-horizontal-1}
	\begin{gathered}
	\mu \norm{\nabla\vec u_{\sigma,++-},\nabla \nabla_h \vec u_{\sigma,++-}}{\Lnorm{2}}^2 + (\mu+\lambda) \norm{\dv \vec u_{\sigma,++-},\nabla_h \dv \vec u_{\sigma,++-}}{\Lnorm{2}}^2 \\ \leq \sum_{\pmindex} \norm{\mathcal A_{\sigma,\mrm j}}{\Lnorm{2}}\norm{\vec u_{\sigma,\mrm j},\nabla_h^2 \vec u_{\sigma,\mrm j}}{\Lnorm{2}}.
	\end{gathered}  
\end{equation}
Notice, the identity \eqref{id:viscosity-tensor} holds for $ \vec u_{\sigma,\pm} $ replaced by $ \vec u_{\sigma,++-} $. Thus, one can derive, straightforwards, that
\begin{gather*}
	\norm{\nabla^2 \vec u_{\sigma,++-}}{\Lnorm{2}} \leq C \norm{\nabla \nabla_h \vec u_{\sigma,++-}}{\Lnorm{2}}\\
	 + \norm{\mu\Delta \vec u_{\sigma,++-} +(\mu+\lambda) \nabla\dv \vec u_{\sigma,++-}}{\Lnorm{2}}.
\end{gather*}
Hence \eqref{est:flat-horizontal-1} implies 
\begin{equation}\label{est:flat-h2-++-}
\norm{\vec u_{\sigma,++-}}{\Hnorm{2}}\leq C \sum_{\pmindex} \bigl( \norm{\mathcal A_{\sigma,\mrm j}}{\Lnorm{2}} + \norm{\vec u_{\sigma,\mrm j}}{\Lnorm{2}} \bigr).
\end{equation}

Similarly, after testing \eqref{eq:flat+--} with $ - \vec u_{\sigma,+--} $, $ \Delta_h \vec u_{\sigma,+--} $, respectively, and applying integration by parts, H\"older's inequality, and Cauchy's inequality in the resultant, one can derive
\begin{equation}\label{est:flat-horizontal-2}
\begin{gathered}
\mu \norm{\nabla\vec u_{\sigma,+--},\nabla \nabla_h \vec u_{\sigma,+--}}{\Lnorm{2}}^2 + (\mu+\lambda) \norm{\dv \vec u_{\sigma,+--},\nabla_h \dv \vec u_{\sigma,+--}}{\Lnorm{2}}^2 \\ 
+ \sigma \int \rho_{\sigma,+} \rho_{\sigma,-} \bigl( 2|\vec u_{\sigma,+--}|^2 +  |\nabla_h \vec u_{\sigma,+--}|^2 \bigr)\idx 
\\
\leq \sum_{\pmindex} \norm{\mathcal A_{\sigma,\mrm j}}{\Lnorm{2}}\norm{\vec u_{\sigma,\mrm j},\nabla_h^2 \vec u_{\sigma,\mrm j}}{\Lnorm{2}} \\
 + 4 \sigma \norm{\nabla_h(\rho_{\sigma,+}\rho_{\sigma,-})^{1/2}}{\Lnorm{6}}^2 \norm{\vec u_{\sigma,+--}}{\Lnorm{3}}^2,
\end{gathered}  
\end{equation}
where we have applying the following estimates,
\begin{align*}
	& - 2 \sigma \int \rho_{\sigma, +}\rho_{\sigma, -} \vec u_{\sigma,+--} \cdot \Delta_h \vec u_{\sigma,+--} \idx \\
	& \quad = 2 \sigma \int \rho_{\sigma, +}\rho_{\sigma, -} |\nabla_h \vec u_{\sigma,+--}|^2 \idx \\
	& \qquad + 2\sigma \sum_{\partial_h \in \lbrace \partial_x,\partial_y \rbrace }\int \partial_h(\rho_{\sigma, +}\rho_{\sigma, -} ) \vec u_{\sigma,+--} \cdot \partial_h \vec u_{\sigma,+--} \idx\\
	& \quad \geq \sigma \int \rho_{\sigma, +}\rho_{\sigma, -} |\nabla_h \vec u_{\sigma,+--}|^2 \idx \\
	& \qquad - 4 \sigma \norm{\nabla_h(\rho_{\sigma,+}\rho_{\sigma,-})^{1/2}}{\Lnorm{6}}^2 \norm{\vec u_{\sigma,+--}}{\Lnorm{3}}^2.
\end{align*}

In addition, due to \eqref{id:viscosity-tensor} and \eqref{eq:flat+--}, one has
\begin{equation}\label{est:flat-vertical-1}
\begin{aligned}
& \norm{\mu \partial_{zz} \vec v_{\sigma,+--} - 2\sigma \rho_{\sigma,+}\rho_{\sigma,-} \vec v_{\sigma,+--}}{\Lnorm{2}} \\
& \quad + \norm{(2\mu+\lambda)\partial_{zz} w_{\sigma,+--} - 2 \sigma \rho_{\sigma,+}\rho_{\sigma,-} w_{\sigma,+--}}{\Lnorm{2}}\\
& \leq C\norm{\nabla \nabla_h \vec u_{\sigma,+--}}{\Lnorm{2}} +  \sum_{\pmindex} \norm{\mathcal A_{\sigma,\mrm j}}{\Lnorm{2}}.
\end{aligned}
\end{equation}
Directly applying integration by parts yields
\begin{gather}
	\begin{aligned} & \norm{\mu \partial_{zz} \vec v_{\sigma,+--} - 2\sigma \rho_{\sigma,+}\rho_{\sigma,-} \vec v_{\sigma,+--}}{\Lnorm{2}}^2 
	= \mu^2 \norm{\partial_{zz} \vec v_{\sigma,+--}}{\Lnorm{2}}^2 \\ & \qquad + 4 \sigma^2 \norm{\rho_{\sigma,+}\rho_{\sigma,-} \vec v_{\sigma,+--}}{\Lnorm{2}}^2
	+ 4 \mu \sigma \int \rho_{\sigma,+} \rho_{\sigma,-} |\partial_z \vec v_{\sigma,+--}|^2 \idx\\
	& \qquad + 4 \mu\sigma \int \partial_z (\rho_{\sigma,+} \rho_{\sigma,-}) \vec v_{\sigma,+--}\cdot \partial_z \vec v_{\sigma,+--} \idx,\end{aligned} \label{est:flat-vertical-2}  \\
	\begin{aligned}
	& \norm{(2\mu+\lambda) \partial_{zz} w_{\sigma,+--} - 2\sigma \rho_{\sigma,+}\rho_{\sigma,-} w_{\sigma,+--}}{\Lnorm{2}}^2 
	\\ & \quad = (2\mu+\lambda)^2 \norm{\partial_{zz} w_{\sigma,+--}}{\Lnorm{2}}^2  + 4 \sigma^2 \norm{\rho_{\sigma,+}\rho_{\sigma,-} w_{\sigma,+--}}{\Lnorm{2}}^2
	\\ & \qquad + 4 (2\mu+\lambda) \sigma \int \rho_{\sigma,+} \rho_{\sigma,-} |\partial_z w_{\sigma,+--}|^2 \idx\\
	& \qquad + 4 (2\mu+\lambda) \sigma \int \partial_z (\rho_{\sigma,+} \rho_{\sigma,-}) w_{\sigma,+--}\cdot \partial_z w_{\sigma,+--} \idx.
	\end{aligned} \label{est:flat-vertical-3}
\end{gather}
Therefore,  \eqref{est:flat-horizontal-2}, \eqref{est:flat-vertical-1}, \eqref{est:flat-vertical-2}, and \eqref{est:flat-vertical-3} yield, after applying H\"older's inequality and Cauchy's inequality,
\begin{equation}\label{est:flat-h2-+--}
	\begin{aligned}
	& \norm{\vec u_{\sigma,+--}}{\Hnorm{2}}^2 + \sigma \int \rho_{\sigma,+}\rho_{\sigma,-} \bigl( |\vec u_{\sigma,+--}|^2 + |\nabla \vec u_{\sigma,+--}|^2 \bigr) \idx \\
	& \qquad + \sigma^2 \norm{\rho_{\sigma, +}\rho_{\sigma, -}\vec u_{\sigma,+--}}{\Lnorm{2}}^2 \leq C \sum_{\pmindex} \bigl( \norm{\mathcal A_{\sigma,\mrm j}}{\Lnorm{2}}^2 + \norm{\vec u_{\sigma,\mrm j}}{\Lnorm{2}}^2 \bigr) \\
	& \qquad +  C \sigma \norm{\nabla_h(\rho_{\sigma,+}\rho_{\sigma,-})^{1/2}}{\Lnorm{6}}^2 \norm{\vec u_{\sigma,+--}}{\Lnorm{3}}^2.
	\end{aligned}
\end{equation}

To finish this subsection, we write down the estimates of $ \norm{\mathcal A_{\sigma,\pm}}{\Lnorm{2}} $.
In fact,
\begin{align*}
	& \norm{\mathcal A_{\sigma,\pm}}{\Lnorm{2}} \leq  \norm{\rho_{\sigma,\pm}}{\Lnorm{\infty}}^{1/2} \norm{\rho_{\sigma,\pm}^{1/2}\dt \vec u_{\sigma,\pm}}{\Lnorm{2}} \\
	& \qquad  + \norm{\rho_{\sigma, \pm}}{\Lnorm{\infty}} \norm{\vec u_{\sigma,\pm}}{\Lnorm{6}} \norm{\nabla \vec u_{\sigma,\pm}}{\Lnorm{3}} + \norm{\nabla p_{\sigma,\pm}}{\Lnorm{2}} \\
	&  \leq \norm{\rho_{\sigma,\pm}}{\Lnorm{\infty}}^{1/2} \norm{\rho_{\sigma,\pm}^{1/2}\dt \vec u_{\sigma,\pm}}{\Lnorm{2}} + \norm{\nabla p_{\sigma,\pm}}{\Lnorm{2}} \\
	& \qquad + \norm{\rho_{\sigma, \pm}}{\Lnorm{\infty}} \bigl( \norm{\rho_{\sigma,\pm}^{1/2}\vec u_{\sigma,\pm}}{\Lnorm{2}} +  \norm{\nabla \vec u_{\sigma,\pm}}{\Lnorm{2}} \bigr) \\
	& \qquad \qquad \times \norm{\nabla \vec u_{\sigma,\pm}}{\Lnorm{2}}^{1/2} \norm{\vec u_{\sigma,\pm}}{\Hnorm{2}}^{1/2}.
\end{align*}
Consequently, combining \eqref{est:flat-h2-++-} and \eqref{est:flat-h2-+--} yields
\begin{equation}\label{est:flat-h2}
	\begin{aligned}
		& \norm{\vec u_{\sigma,\pm}}{\Hnorm{2}}^2 + \sigma^2 \norm{\rho_{\sigma, +}\rho_{\sigma, -}(\vec u_{\sigma,+} - \vec u_{\sigma,-})}{\Lnorm{2}}^2  \\
		& \qquad  + \sigma \int \rho_{\sigma,+}\rho_{\sigma,-} \bigl( |\vec u_{\sigma,+} - \vec u_{\sigma,-}|^2 + |\nabla (\vec u_{\sigma,+} - \vec u_{\sigma,-})|^2 \bigr) \idx \\
		& \quad \leq \mathcal H(\norm{\rho_{\sigma,\pm}}{\Lnorm{\infty}}, \norm{\nabla \rho_{\sigma,\pm}}{\Lnorm{2}},
		\norm{\rho_{\sigma,\pm}^{1/2}\vec u_{\sigma,\pm}}{\Lnorm{2}}, \norm{\nabla\vec u_{\sigma,\pm}}{\Lnorm{2}}, \\
		&\qquad \qquad \norm{\rho_{\sigma, \pm}^{1/2}\dt \vec u_{\sigma,\pm}}{\Lnorm{2}}) \\
		& \qquad
		+  C \sigma \biggl( 1 + \sup_{\vec x \in \Omega} \dfrac{1}{\rho_{\sigma, +}\rho_{\sigma, -}} \biggr) \norm{\nabla_h(\rho_{\sigma,+}\rho_{\sigma,-})^{1/2}}{\Lnorm{6}}^2 \norm{\vec u_{\sigma,+--}}{\Lnorm{2}}^2,
	\end{aligned}
\end{equation}
where we have applied
\begin{equation*}\label{ineq:interpolation-2}
\norm{\vec u_{\sigma,+--}}{\Lnorm{3}} \leq C  \norm{\vec u_{\sigma,+--}}{\Lnorm{2}}^{1/2} 
\bigl( \norm{\vec u_{\sigma,+--}}{\Lnorm{2}}^{1/2} + \norm{\nabla \vec u_{\sigma,+--}}{\Lnorm{2}}^{1/2} \bigr).
\end{equation*}

\subsection{Normal derivative estimates: $H^3$}

Recall that $ \partial_h \in \lbrace \partial_x,\partial_y \rbrace $. Applying $\partial_h$ to \eqref{eq:flat++-} and \eqref{eq:flat+--} leads to
\begin{gather*}
	\mu\Delta \partial_h \vec u_{\sigma,++-} +(\mu+\lambda) \nabla\dv \partial_h \vec u_{\sigma,++-} = \partial_h \mathcal A_{\sigma,+} + \partial_h \mathcal A_{\sigma,-},  \\
	\mu\Delta \partial_h \vec u_{\sigma,+--}  +(\mu+\lambda) \nabla\dv \partial_h \vec u_{\sigma,+--} - 2 \sigma \rho_{\sigma,+}\rho_{\sigma,-} \partial_h\vec u_{\sigma,+--}  \\
	= \partial_h\mathcal A_{\sigma,+} - \partial_h \mathcal A_{\sigma,-} + 2 \sigma \partial_h (\rho_{\sigma,+}\rho_{\sigma,-}) \vec u_{\sigma,+--}.
\end{gather*}
Then, by performing similar estimates from \eqref{est:flat-horizontal-1} to \eqref{est:flat-h2-+--}, one can obtain the following estimates:
\begin{gather}
	\norm{\nabla_h \vec u_{\sigma,++-}}{\Hnorm{2}}\leq C \sum_{\pmindex} \bigl( \norm{\nabla_h \mathcal A_{\sigma,\mrm j}}{\Lnorm{2}} + \norm{\nabla_h \vec u_{\sigma,\mrm j}}{\Lnorm{2}} \bigr)
	\label{est:flat-h3-++--01}, \\
	\begin{aligned}
	& \norm{\nabla_h \vec u_{\sigma,+--}}{\Hnorm{2}}^2 + \sigma \int \rho_{\sigma,+}\rho_{\sigma,-} \bigl( |\nabla_h \vec u_{\sigma,+--}|^2 + |\nabla_h \nabla \vec u_{\sigma,+--}|^2 \bigr) \idx \\
	& \qquad + \sigma^2 \norm{\rho_{\sigma, +}\rho_{\sigma, -}\nabla_h \vec u_{\sigma,+--}}{\Lnorm{2}}^2 \leq C \sum_{\pmindex} \bigl( \norm{\nabla_h \mathcal A_{\sigma,\mrm j}}{\Lnorm{2}}^2 \\
	& \qquad \qquad \qquad + \norm{\nabla_h \vec u_{\sigma,\mrm j}}{\Lnorm{2}}^2 \bigr) \\
	& \qquad +  C \sigma \norm{\nabla_h(\rho_{\sigma,+}\rho_{\sigma,-})^{1/2}}{\Lnorm{6}}^2 \norm{\nabla_h \vec u_{\sigma,+--}}{\Lnorm{3}}^2 \\
	& \qquad + C \sigma^2 \norm{\nabla_h(\rho_{\sigma, +}\rho_{\sigma, -})}{\Lnorm{6}}^2 \norm{\vec u_{\sigma,+--}}{\Lnorm{3}}^2.
	\end{aligned}\label{est:flat-h3-+---01}
\end{gather}

On the other hand, after applying $ \partial_z $ to \eqref{eq:flat++-}, together with identity \eqref{id:viscosity-tensor}, it follows:
\begin{align*}
	& \left(\begin{array}{c}
	\mu \partial_{zzz} \vec v_{\sigma,++-} \\
	(2\mu+\lambda) \partial_{zzz} w_{\sigma,++-} 
	\end{array} \right) 
	 =  \partial_z \mathcal A_{\sigma,+} + \partial_z  \mathcal A_{\sigma,-} \\
	 & \qquad - \mu \Delta_h \partial_z \vec u_{\sigma,++-} - (\mu+\lambda) \left(\begin{array}{c}
	\nabla_h\dv \partial_z \vec u_{\sigma,++-}\\
	\partial_z \dvh \partial_{z} \vec v_{\sigma,++-}
	\end{array}\right).
\end{align*}
Therefore, it holds:
\begin{equation}\label{est:flat-h3-++-02}
	\norm{\partial_{zzz}\vec u_{\sigma,++-}}{\Lnorm{2}} \leq C \sum_{\pmindex} \norm{\partial_z \mathcal A_{\sigma,\mrm j}}{\Lnorm{2}}
	 + C \norm{\nabla_h \vec u_{\sigma,++-}}{\Hnorm{2}}. 
\end{equation}

Similarly, after applying $ \partial_z $ to \eqref{eq:flat+--}, together with identity \eqref{id:viscosity-tensor}, it follows:
\begin{equation}\label{id:flat-h3-vertical}
	\begin{aligned}
	& \left(\begin{array}{c}
	\mu \partial_{zzz} \vec v_{\sigma,+--} - 2 \sigma \rho_{\sigma, +}\rho_{\sigma, -} \partial_z \vec v_{\sigma,+--} \\
	(2\mu+\lambda) \partial_{zzz} w_{\sigma,+--} - 2 \sigma \rho_{\sigma,+}\rho_{\sigma, -} \partial_z w_{\sigma,+--}
	\end{array} \right) \\
	 & =  \partial_z \mathcal A_{\sigma,+} - \partial_z  \mathcal A_{\sigma,-} + 2 \sigma  \partial_z (\rho_{\sigma, +}\rho_{\sigma, -}) \vec u_{\sigma,+--} \\
	 & \qquad  - \mu \Delta_h \partial_z \vec u_{\sigma,+--} - (\mu+\lambda) \left(\begin{array}{c}
	  \nabla_h\dv \partial_z \vec u_{\sigma,+--}\\
	  \partial_z \dvh \partial_{z} \vec v_{\sigma,+--}
	  \end{array}\right).
	\end{aligned}
\end{equation}
Again, directly applying integration by parts yields
\begin{gather}
	\begin{aligned}
	& \norm{\mu \partial_{zzz} \vec v_{\sigma,+--} - 2\sigma \rho_{\sigma,+}\rho_{\sigma,-} \partial_z \vec v_{\sigma,+--}}{\Lnorm{2}}^2 
	= \mu^2 \norm{\partial_{zzz} \vec v_{\sigma,+--}}{\Lnorm{2}}^2 \\ & \qquad + 4 \sigma^2 \norm{\rho_{\sigma,+}\rho_{\sigma,-}\partial_z \vec v_{\sigma,+--}}{\Lnorm{2}}^2
	+ 4 \mu \sigma \int \rho_{\sigma,+} \rho_{\sigma,-} |\partial_{zz} \vec v_{\sigma,+--}|^2 \idx\\
	& \qquad + 4 \mu\sigma \int \partial_z (\rho_{\sigma,+} \rho_{\sigma,-}) \partial_z \vec v_{\sigma,+--}\cdot \partial_{zz} \vec v_{\sigma,+--} \idx,
	\end{aligned} \label{est:flat-vertical-4} \\
	\begin{aligned}
	& \norm{(2\mu+\lambda) \partial_{zzz} w_{\sigma,+--} - 2\sigma \rho_{\sigma,+}\rho_{\sigma,-} \partial_z w_{\sigma,+--}}{\Lnorm{2}}^2 
	\\ & \quad = (2\mu+\lambda)^2 \norm{\partial_{zzz} w_{\sigma,+--}}{\Lnorm{2}}^2  + 4 \sigma^2 \norm{\rho_{\sigma,+}\rho_{\sigma,-} \partial_z w_{\sigma,+--}}{\Lnorm{2}}^2
	\\ & \qquad + 4 (2\mu+\lambda) \sigma \int \rho_{\sigma,+} \rho_{\sigma,-} |\partial_{zz} w_{\sigma,+--}|^2 \idx\\
	& \qquad + 4 (2\mu+\lambda) \sigma \int \partial_z (\rho_{\sigma,+} \rho_{\sigma,-}) \partial_{z} w_{\sigma,+--}\cdot \partial_{zz} w_{\sigma,+--} \idx \\
	& \qquad - 4 (2\mu+\lambda) \sigma \biggl(\int_{\mathbb T^2}  \rho_{\sigma,+} \rho_{\sigma,-} \partial_{z} w_{\sigma,+--}\cdot \partial_{zz} w_{\sigma,+--} \,dS\biggr) \bigg|_{z = 0}^{1}. 
	\end{aligned} \label{est:flat-vertical-5}
\end{gather}
Moreover, applying the trace embedding inequality and the Gagliardo-Nirenberg inequality implies
\begin{equation}\label{est:flat-vertical-6}
\begin{aligned}
	& \biggl(\int_{\mathbb T^2}  \rho_{\sigma,+} \rho_{\sigma,-} \partial_{z} w_{\sigma,+--}\cdot \partial_{zz} w_{\sigma,+--} \,dS\biggr) \bigg|_{z = 0}^{1} \leq \norm{\rho_{\sigma,+}\rho_{\sigma,-}}{\bLnorm{\infty}} \\
	& \qquad \times \norm{\partial_{z}w_{\sigma,+--}}{\bLnorm{2}}\norm{\partial_{zz} w_{\sigma,+--}}{\bLnorm{2}} \leq \norm{\rho_{\sigma,+}\rho_{\sigma,-}}{\bHnorm{3/2}} \\
	& \qquad \times \norm{\partial_{z}w_{\sigma,+--}}{\Hnorm{1/2}}\norm{\partial_{zz} w_{\sigma,+--}}{\Hnorm{1/2}} \leq C \norm{\rho_{\sigma,+}\rho_{\sigma,-}}{\Hnorm{2}} \\
	& \qquad \times \norm{\partial_{z}w_{\sigma,+--}}{\Lnorm{2}}^{1/2} \norm{\partial_z w_{\sigma,+--}}{\Hnorm{1}}^{1/2} \norm{\partial_{zz} w_{\sigma,+--}}{\Lnorm{2}}^{1/2}\\
	& \qquad \times \norm{\partial_{zz} w_{\sigma,+--}}{\Hnorm{1}}^{1/2}.
\end{aligned}
\end{equation}
Therefore, combining \eqref{id:flat-h3-vertical}--\eqref{est:flat-vertical-6} leads to
\begin{equation}\label{est:flat-h3-+--02}
	\begin{aligned}
		& \norm{\partial_{zzz} \vec u_{\sigma,+--} }{\Lnorm{2}}^2 + \sigma^2 \norm{\rho_{\sigma, +}\rho_{\sigma, -}\partial_z \vec u_{\sigma,+--}}{\Lnorm{2}}^2 \\
		& \qquad + \sigma \int \rho_{\sigma, +}\rho_{\sigma, -}|\partial_{zz}\vec u_{\sigma,+--}|^2 \idx
		 \leq C \sum_{\pmindex}\norm{\partial_z \mathcal A_{\sigma,\mrm j}}{\Lnorm{2}}^2 \\
		&\qquad + C \norm{\nabla_h \vec u_{\sigma,+--}}{\Hnorm{2}}^2 
		+ C \sigma \norm{\partial_{z}(\rho_{\sigma, +}\rho_{\sigma, -})^{1/2}}{\Lnorm{6}}^2 \norm{\partial_{z}\vec u_{\sigma,+--}}{\Lnorm{3}}^2 \\
		& \qquad + C \sigma^2 \norm{\partial_z(\rho_{\sigma, +}\rho_{\sigma, -})}{\Lnorm{6}}^2 \norm{\vec u_{\sigma,+--}}{\Lnorm{3}}^2 \\
		& \qquad + \mathcal H(\norm{\rho_{\sigma,\pm}}{\Hnorm{2}},\norm{\vec u_{\sigma,\pm}}{\Hnorm{2}}).
	\end{aligned}
\end{equation}

To finish this subsection, we write down the estimates of $ \norm{\nabla \mathcal A_{\sigma,\pm}}{\Lnorm{2}} $. Direct calculation yields,
\begin{align*}
	& \norm{\nabla \mathcal A_{\sigma,\pm}}{\Lnorm{2}} \leq C \norm{\nabla \rho_{\sigma,\pm}^{1/2}}{\Lnorm{6}} \norm{\rho_{\sigma,\pm}^{1/2} \dt \vec u_{\sigma,\pm}}{\Lnorm{3}} \\
	& \qquad + C \norm{\rho_{\sigma,\pm}}{\Lnorm{\infty}} \norm{\nabla\dt \vec u_{\sigma,\pm}}{\Lnorm{2}} \\
	& \qquad + C \norm{\nabla \rho_{\sigma,\pm}}{\Lnorm{6}}\norm{\vec u_{\sigma,\pm}}{\Lnorm{\infty}} \norm{\nabla\vec u_{\sigma,\pm}}{\Lnorm{3}} \\
	& \qquad + C \norm{\rho_{\sigma, \pm}}{\Lnorm{\infty}} \norm{\nabla \vec u_{\sigma,\pm}}{\Lnorm{6}} \norm{\nabla\vec u_{\sigma,\pm}}{\Lnorm{3}}\\
	& \qquad + C \norm{\rho_{\sigma, \pm}}{\Lnorm{\infty}} \norm{ \vec u_{\sigma,\pm}}{\Lnorm{\infty}} \norm{\nabla^2\vec u_{\sigma,\pm}}{\Lnorm{2}} + C \norm{\nabla^2 p_{\sigma,\pm}}{\Lnorm{2}}.
\end{align*}
Consequently, combining \eqref{est:flat-h3-++--01}, \eqref{est:flat-h3-+---01}, \eqref{est:flat-h3-++-02}, and \eqref{est:flat-h3-+--02} yields
\begin{equation}\label{est:flat-h3}
	\begin{aligned}
	& \norm{\vec u_{\sigma,\pm}}{\Hnorm{3}}^2 + \sigma^2 \norm{\rho_{\sigma, +}\rho_{\sigma, -}\nabla(\vec u_{\sigma,+}-\vec u_{\sigma,-})}{\Lnorm{2}}^2 \\
	& \qquad + \sigma \int \rho_{\sigma, +}\rho_{\sigma, -} | \nabla^2 (\vec u_{\sigma,+} - \vec u_{\sigma,-})|^2 \idx \\
	& \leq C \sigma \int \rho_{\sigma, +}\rho_{\sigma, -} | \nabla (\vec u_{\sigma,+} - \vec u_{\sigma,-})|^2 \idx \\
	& \qquad + C \sigma \biggl(1+\sup_{\vec x\in\Omega}\dfrac{1}{\rho_{\sigma,+}\rho_{\sigma, -}} \biggr) \norm{\nabla(\rho_{\sigma,+}\rho_{\sigma,-})^{1/2}}{\Lnorm{6}}^2 \norm{\nabla (\vec u_{\sigma,+}-\vec u_{\sigma,-})}{\Lnorm{2}}^2 \\
	& \qquad + C \sigma^2 \biggl(1+\sup_{\vec x\in\Omega}\dfrac{1}{(\rho_{\sigma,+}\rho_{\sigma, -})^2} \biggr)  \norm{\nabla(\rho_{\sigma, +}\rho_{\sigma, -})}{\Lnorm{6}}^2 \norm{\vec u_{\sigma,+}-\vec u_{\sigma,-}}{\Lnorm{2}}^2 \\
	& \qquad + \sum_{\pmindex}( \norm{\rho_{\sigma,\mrm j}}{\Lnorm{\infty}}^2 + 1) \norm{\nabla \dt \vec u_{\sigma,\mrm j}}{\Lnorm{2}}^2 \\
	& \qquad + \mathcal H(\norm{\nabla\rho_{\sigma,\pm}^{1/2}}{\Lnorm{6}},\norm{\rho_{\sigma,\pm}}{\Hnorm{2}},\norm{\vec u_{\sigma,\pm}}{\Hnorm{2}},\norm{\rho_{\sigma,\pm}^{1/2}\dt \vec u_{\sigma,\pm}}{\Lnorm{2}}),
	\end{aligned}
\end{equation}
where we have applied interpolation inequality \eqref{ineq:interpolation-2} and
\begin{equation*}
\norm{\nabla \vec u_{\sigma,+--}}{\Lnorm{3}} \leq C  \norm{\nabla \vec u_{\sigma,+--}}{\Lnorm{2}}^{1/2} 
\bigl( \norm{\nabla \vec u_{\sigma,+--}}{\Lnorm{2}}^{1/2} + \norm{\nabla^2 \vec u_{\sigma,+--}}{\Lnorm{2}}^{1/2} \bigr).
\end{equation*}

\subsection{Uniform-in-$ \sigma $ estimates}\label{subsec:uniform-2}

In this subsection, we summarize the estimates above and derive the uniform-in-$ \sigma $ estimates. Again, let $ T>0 $ be the uniform-in-$\sigma$ life span of solutions to \eqref{eq:2CMPNS}, and denote $ \mathfrak M $ as in \eqref{def:bound-dsptn-cauchy}, which will  be shown to be independent of $ \sigma $. Also, $ C_{\mrm{in}} \in (0,\infty) $ represents a generic constant depending only on
\begin{gather*}
\underline{\rho} > 0, ~ \norm{\nabla \log \rho_{\sigma,\pm,0}}{\Lnorm{3}}, ~ \norm{\rho_{\sigma,\pm,0}^{1/2}}{\Hnorm{2}}, ~ \norm{\vec u_{\sigma,\pm,0}}{\Hnorm{1}}, ~ \norm{\vec u_{\sigma,\pm,1}}{\Lnorm{2}},
\end{gather*} 
but independent of $ \sigma $, which might be different from line to line, where $ \vec u_{\sigma, \pm, 1} = \dt \vec u_{\sigma,\pm}\big|_{t=0} $ are the initial data of $ \dt \vec u_{\sigma,\pm} $ defined by equation \subeqref{eq:2CMPNS}{2}.

Then following the same lines of proof, one can establish estimate \eqref{est:rho-total-1}, after choosing $ T $ small enough.

On the other hand, for $ \sigma \in (\sigma^*,\infty) $, with some $\sigma^* $ large enough, depending on $ C_{\mrm{in}} $, such that in \eqref{est:flat-h2},
\begin{gather*}
C \sigma \biggl( 1 + \sup_{\vec x \in \Omega} \dfrac{1}{\rho_{\sigma, +}\rho_{\sigma, -}} \biggr) \norm{\nabla_h(\rho_{\sigma,+}\rho_{\sigma,-})^{1/2}}{\Lnorm{6}}^2 \norm{\vec u_{\sigma,+--}}{\Lnorm{2}}^2\\
 \leq \dfrac{\sigma^2}{2} \norm{\rho_{\sigma, +}\rho_{\sigma, -}(\vec u_{\sigma,+}-\vec u_{\sigma,-})}{\Lnorm{2}}^2.
\end{gather*}
That is, for $ \sigma \geq \sigma^* $, \eqref{est:flat-h2} implies
\begin{equation}\label{est:flat-h2-final}
\begin{aligned}
& \norm{\vec u_{\sigma,\pm}}{\Hnorm{2}}^2 + \sigma^2 \norm{\rho_{\sigma, +}\rho_{\sigma, -}(\vec u_{\sigma,+} - \vec u_{\sigma,-})}{\Lnorm{2}}^2  \\
& \qquad  + \sigma \int \rho_{\sigma,+}\rho_{\sigma,-} \bigl( |\vec u_{\sigma,+} - \vec u_{\sigma,-}|^2 + |\nabla (\vec u_{\sigma,+} - \vec u_{\sigma,-})|^2 \bigr) \idx \\
& \quad \leq \mathcal H(\norm{\rho_{\sigma,\pm}}{\Lnorm{\infty}}, \norm{\nabla \rho_{\sigma,\pm}}{\Lnorm{2}},
\norm{\rho_{\sigma,\pm}^{1/2}\vec u_{\sigma,\pm}}{\Lnorm{2}}, \norm{\nabla\vec u_{\sigma,\pm}}{\Lnorm{2}}, \\
&\qquad \qquad \norm{\rho_{\sigma, \pm}^{1/2}\dt \vec u_{\sigma,\pm}}{\Lnorm{2}}). \\
\end{aligned}
\end{equation}

Then after collecting \eqref{est:002} and \eqref{est:001-flat}, substituting \eqref{ineq:interpolation-1} and \eqref{est:flat-h2-final}, and applying Young's inequality, one has
\begin{equation}\label{energy-est:total-flat}
	\begin{aligned}
	& \dfrac{d}{dt} E_{\sigma,2} + D_{\sigma,2} + \sigma \int \rho_{\sigma,+} \rho_{\sigma,-} |\dt (\vec u_{\sigma,+}-\vec u_{\sigma,-})|^2  \idx \\
	& \quad \leq \mathfrak c_3 \sigma \int \rho_{\sigma,+} \rho_{\sigma,-} |\vec u_{\sigma,+}-\vec u_{\sigma,-}|^2 \idx + \mathcal H( E_\sigma, E_{\sigma,2}, C_\mrm{in} ),
	\end{aligned}
	\end{equation}
	where $ \mathfrak c_3 \in (0,\infty) $ is a constant depending on $ C_\mrm{in} $, but independent of $ \sigma $, and
	\begin{align}
	& \begin{aligned} & E_{\sigma,2} := \sum_{\pmindex} \biggl\lbrace \norm{\rho_{\sigma,\mrm j}^{1/2} \dt \vec u_{\sigma,\mrm j}}{\Lnorm{2}}^2 + \norm{\nabla \vec u_{\sigma,\mrm j}}{\Lnorm{2}}^2\biggr\rbrace, \end{aligned} \\
	& \begin{aligned} & D_{\sigma,2} : = \sum_{\pmindex} \bigg\lbrace \mu_\mrm j \norm{\nabla \dt \vec u_{\sigma,\mrm j}}{\Lnorm{2}}^2 +  (\lambda_\mrm j + \mu_+\mrm j) \norm{\dv \dt \vec u_{\sigma,\mrm j}}{\Lnorm{2}}^2 
	\biggr\rbrace.
	\end{aligned}
	\end{align}
Consequently, \eqref{cnsvt:energy} and \eqref{energy-est:total-flat} yield that, there is some $ T^{**} \in (0,T] $, independent of $ \sigma $, such that for any $ t \in (0,T^{**}) $,
\begin{equation}\label{est:total-final-flat}
	\begin{aligned}
	& \sup_{0\leq s\leq t} \bigl\lbrace E_\sigma(s) + E_{\sigma,2}(s) \bigr\rbrace + \int_0^t \bigl\lbrace D_\sigma(s) + D_{\sigma,2}(s) \bigr\rbrace \,ds \\
	& \qquad + \sigma \int_0^t \biggl\lbrace \norm{(\vec u_{\sigma,+}-\vec u_{\sigma,-})(s)}{\Lnorm{2}}^2 + \norm{\dt (\vec u_{\sigma,+}-\vec u_{\sigma,-})(s)}{\Lnorm{2}}^2 \biggr\rbrace \,ds\\
	& \quad\leq C_\mrm{in},
	\end{aligned}
	\end{equation}
	which, in particular, implies, together with \eqref{est:flat-h2-final} and \eqref{est:flat-h3}
	\begin{equation}\label{est:update-bound-flat-01}
	\begin{gathered}
	\sup_{0\leq s\leq T^{**}}\norm{\dt \vec u_{\sigma,\pm}(s)}{\Lnorm{2}}^2+ \int_0^{T^{**}} \biggl\lbrace \norm{\nabla \vec u_{\sigma,\pm}(s)}{\Hnorm{2}}^2 \\ + \norm{\nabla \dt \vec u_{\sigma,\pm}(s)}{\Lnorm{2}}^2 \biggr\rbrace \,ds  \leq \mathfrak M^{**}(C_{\mrm{in}}).\end{gathered}
	\end{equation}
	
Now, while keeping $ \sigma \geq \sigma^* $, we update $ \mathfrak M $ to $ \mathfrak M^{**}(C_\mrm{in}) $, and $ T $ to $ T^{**} $ accordingly. To this end, we conclude this section with \eqref{est:rho-total-1} and the following uniform estimates:
\begin{equation}\label{est:u-total-2}
\begin{aligned}
& \sup_{0\leq s\leq T} \biggl\lbrace \norm{\vec u_{\sigma,\pm}(s)}{\Hnorm{2}}^2 + \norm{\dt \vec u_{\sigma,\pm}(s)}{\Lnorm{2}}^2 + \sigma^2 \norm{(\vec u_{\sigma,+} - \vec u_{\sigma,-})(s)}{\Lnorm{2}}^2 \\
& \qquad + \sigma \norm{(\vec u_{\sigma,+} - \vec u_{\sigma,-})(s)}{\Hnorm{1}}^2 \biggr\rbrace \\
& \quad + \int_0^T \biggl\lbrace \norm{\vec u_{\sigma,\pm}(s)}{\Hnorm{3}}^2 + \norm{\dt \vec u_{\sigma,\pm}(s)}{\Hnorm{1}}^2 \biggr\rbrace \,ds \\
& \quad + \sigma^2 \int_0^T \norm{(\vec u_{\sigma,+} - \vec u_{\sigma,-})(s)}{\Hnorm{1}}^2 \,ds \\
& \quad + \sigma \int_0^T \biggl\lbrace \norm{(\vec u_{\sigma,+}-\vec u_{\sigma,-})(s)}{\Hnorm{2}}^2 +  \norm{\dt (\vec u_{\sigma,+}-\vec u_{\sigma,-})(s)}{\Lnorm{2}}^2 \biggr\rbrace \,ds\\
& \qquad \leq C_{\mrm{in},T}.
\end{aligned}
\end{equation}

\section{Passing the limit $ \sigma \rightarrow \infty $}\label{sec:limit}

We are ready to pass the limit $ \sigma \rightarrow \infty $ in system \eqref{eq:2CMPNS}. From the uniform estimates in \eqref{est:rho-total-1}, \eqref{est:u-total-1}, and \eqref{est:u-total-2} for $ \Omega = \mathbb T^3 $ and $ \Omega = \mathbb{T}^2 \times (0,1) $, respectively, one can conclude that, there exist $ \rho_{\pm} $ and $ \vec u $ in the corresponding spaces, such that
\begin{equation}\label{limit:general-1}
\begin{aligned}
	\vec u_{\sigma,+} - \vec u_{\sigma,-} & \rightarrow 0 & \quad \text{in} \quad & L^\infty(0,T;H^1(\Omega))\cap L^2(0,T;H^2(\Omega)), \\
	\dt (\vec u_{\sigma,+} - \vec u_{\sigma,-})  & \rightarrow 0 & \quad \text{in} \quad & L^2(0,T;L^2(\Omega)),\\
	\rho_{\sigma, \pm} & \rightarrow \rho_{\pm} & \quad \text{in} \quad &  C(0,T;H^1(\Omega)),\\
	\rho_{\sigma, \pm} & \buildrel\ast\over\rightharpoonup \rho_{\pm} & \quad \text{weak-$\star$ in} \quad &  L^\infty(0,T;H^2(\Omega)), \\
	\dt \rho_{\sigma, \pm} & \buildrel\ast\over\rightharpoonup \dt \rho_{\pm} & \quad \text{weak-$\star$ in} \quad & L^\infty(0,T;L^2(\Omega)), \\
	\vec u_{\sigma,\pm} & \rightarrow  \vec u & \quad \text{in} \quad & C(0,T;H^1(\Omega)) \cap L^2(0,T;H^2(\Omega)), \\
	\vec u_{\sigma,\pm} & \rightarrow  \vec u & \quad \text{weakly in} \quad & L^2(0,T;H^3(\Omega)), \\
	\vec u_{\sigma,\pm} & \rightarrow  \vec u & \quad \text{weak-$\star$ in} \quad & L^\infty(0,T;H^2(\Omega)), \\
	\dt \vec u_{\sigma,\pm} & \rightharpoonup  \dt \vec u & \quad \text{weakly in} \quad & L^2(0,T;H^1(\Omega)), \\
	\dt \vec u_{\sigma,\pm} & \buildrel\ast\over\rightharpoonup  \dt \vec u & \quad \text{weak-$\star$ in} \quad & L^\infty(0,T;L^2(\Omega)), 
\end{aligned}
\end{equation}
as $ \sigma \rightarrow \infty $. Meanwhile, after adding the $ + $-equation and the $ - $-equation of \subeqref{eq:2CMPNS}{2} together, one obtains
\begin{equation}\label{eq:conv-1}
	\begin{aligned}
	& \dt (\rho_{\sigma,+} \vec u_{\sigma,+} + \rho_{\sigma, -}\vec u_{\sigma,-}) +\dv(\rho_{\sigma, +}\vec u_{\sigma,+}\otimes \vec u_{\sigma,+} + \rho_{\sigma, -}\vec u_{\sigma,-}\otimes \vec u_{\sigma,-})\\
	& \qquad + \nabla (p_{\sigma,+} + p_{\sigma,-}) = \dv (\mathbb S_{\sigma,+} + \mathbb S_{\sigma,-}).
	\end{aligned}
\end{equation} 
Passing $ \sigma \rightarrow \infty $ and utilizing \eqref{limit:general-1} in \subeqref{eq:2CMPNS}{1} and \eqref{eq:conv-1} lead to the convergences of them to \eqref{eq:limiting-eq} in the sense of distribution.

Thus we have verified the asymptotic limit \eqref{eq:limiting-eq} of system \eqref{eq:2CMPNS}. 

Our remaining goal is to investigate the asymptote of
$$
\sigma \rho_{\sigma,\pm} \rho_{\sigma,\mp} (\vec u_{\sigma,\mp} - \vec u_{\sigma,\pm})
$$
in \subeqref{eq:2CMPNS}{2}, which, in general, does not vanish as $ \sigma \rightarrow \infty $. 
Notice that, from \eqref{eq:+--}, we have
\begin{equation}\label{eq:conv-2}
\begin{aligned}
& \sigma( \vec u_{\sigma,-} - \vec u_{\sigma,+}) = \dfrac{1}{ \rho_{\sigma, +} + \rho_{\sigma, -}}\biggl( \dt (\vec u_{\sigma,+} - \vec u_{\sigma,-}) + \vec u_{\sigma,+} \cdot \nabla \vec u_{\sigma,+} \\
& \qquad - \vec u_{\sigma,-} \cdot \nabla \vec u_{\sigma,-} 
 + \dfrac{\nabla p_{\sigma,+} - \dv \mathbb S_{\sigma,+}}{\rho_{\sigma,+}} - \dfrac{\nabla p_{\sigma,-} - \dv \mathbb S_{\sigma,-}}{\rho_{\sigma,-}} \biggr).
\end{aligned}
\end{equation}
Therefore, passing $ \sigma \rightarrow \infty $ in \eqref{eq:conv-2} implies that, in the sense of distribution,
\begin{equation}\label{limit:001}
	\begin{aligned}
	& \sigma( \vec u_{\sigma,-} - \vec u_{\sigma,+}) \buildrel\text{$\sigma\rightarrow\infty$}\over\rightharpoonup \dfrac{1}{\rho_+ + \rho_-} \biggl\lbrack \dfrac{R_+ \gamma_+ }{\gamma_+ - 1} \nabla \rho_+^{\gamma_+ - 1} - \dfrac{R_- \gamma_- }{\gamma_- - 1} \nabla \rho_-^{\gamma_- - 1}  \\
	& \qquad + \biggl( \dfrac{\mu_+}{\rho_+} - \dfrac{\mu_-}{\rho_-}\biggr)\dv (\nabla \vec u + \nabla^\top \vec u) + \biggl( \dfrac{\lambda_+}{\rho_+} - \dfrac{\lambda_-}{\rho_-}\biggr) \nabla \dv \vec u \biggr\rbrack = \mathfrak S.
	\end{aligned}
\end{equation}
Consequently, after passing $ \sigma \rightarrow \infty $ in \subeqref{eq:2CMPNS}{2} in the sense of distribution, one verifies \eqref{eq:conv-3}. 

\bibliographystyle{plain}

\end{document}